\documentclass[letter,11pt]{article}
\usepackage{multirow}
\usepackage{lineno}
\usepackage{microtype}%if unwanted, comment out or use option "draft"
\usepackage{epsfig}
\usepackage{microtype}
\usepackage{graphicx}

\usepackage[linesnumbered,ruled, lined]{algorithm2e}
\usepackage{algpseudocode}
\usepackage{epsfig,enumerate,amsmath,amsfonts,amssymb,amsthm,mathrsfs,ifpdf}
\usepackage{indentfirst,relsize}
\usepackage{setspace}
\usepackage{enumerate}
\usepackage{enumitem}
\usepackage{latexsym}
\usepackage{stackrel}
\usepackage[all]{xy}
\usepackage[margin = 2.5cm]{geometry}
\usepackage[usenames,dvipsnames]{pstricks}
\usepackage{pst-grad} % For gradients
\usepackage{pst-plot} % For axes
\usepackage{xspace}
\newcommand{\size}[1]{\left| #1 \right|}

\newcommand{\E}{\mathbb{E}}

\newcommand{\remove}[1]{}

\newcommand{\poly}{\text{poly}}

\newcommand{\N}{\mathbb{N}}

\newcommand{\cT}{\mathcal{T}}

\newcommand{\cE}{\mathcal{E}}
\newcommand{\cL}{\mathcal{L}}
\newcommand{\cA}{\mathcal{A}}

\newcommand{\cD}{\mathcal{D}}

\newcommand{\Oh}{\mathcal{O}}
\newcommand{\tOh}{\widetilde{\mathcal{O}}}
\newcommand{\tOm}{\widetilde{\Omega}}

\newcommand{\cH}{\mathcal{H}}

\newcommand{\pr}{\mathbb{P}}
\newcommand{\eps}{\epsilon}
\newcommand{\vareps}{\varepsilon}

\newcommand{\complain}[1]{\textcolor{red}{#1}}
\newcommand{\comments}[1]{\textcolor{blue}{\bf{#1}}}

\newcommand{\etal}{{\it{et al.}}}

\theoremstyle{plain}
\newtheorem{theo}{Theorem}[section]
\newtheorem{lem}[theo]{Lemma}
\newtheorem{pro}[theo]{Proposition}
\newtheorem{coro}[theo]{Corollary}

\newtheorem{cl}[theo]{Claim}

\theoremstyle{definition}
\newtheorem{defi}[theo]{Definition}
\newtheorem{rem}{Remark}
\newtheorem{obs}[theo]{Observation}

\newcommand{\ee}{{\sc EE}\xspace}
\newcommand{\eelong}{{\sc Edge Emptiness}\xspace}
\newcommand{\bis}{{\sc BIS}\xspace}
\newcommand{\bislong}{{\sc Bipartite Independent Set}\xspace}
\newcommand{\is}{{\sc IS}\xspace}
\newcommand{\islong}{{\sc Independent Set}\xspace}
\newcommand{\tis}{{\sc TIS}\xspace}
\newcommand{\tislong}{{\sc Tripartite Independent Set}\xspace}
\newcommand{\gpis}{{\sc CID}\xspace}
\newcommand{\gpislong}{{\sc Colorful Independence Oracle}\xspace}

\newcommand{\yesdist}{{\bf {\cal D}_{Yes}}}
\newcommand{\nodist}{{\bf {\cal D}_{No}}}
\newcommand{\alg}{\mbox{{\sc ALG}}}
\newcommand{\noleaf}{{\bf \cL_{No}}}

\newcommand{\staree}{$\mbox{\ee}^*$\xspace}
\newcommand{\staralg}{$\mbox{{\sc ALG}}^*$\xspace}

\newcommand{\alghigh}{\mbox{{\sc Triangle-Est-High}} \xspace}
\newcommand{\alglow}{\mbox{{\sc Triangle-Est-Low}}\xspace}
\newcommand{\test}{\mbox{{\sc Triangle-Est}}\xspace}

\newcommand{\type}{\mbox{{\sc type}-}}
\renewcommand{\eps}{\vareps}
\usepackage[symbol]{footmisc}
\usepackage{hyperref}
\begin{document}
\title{{\bf On the Complexity of Triangle Counting using
Emptiness Queries}
}
{
\author{
Arijit Bishnu \footnote{Indian Statistical Institute, Kolkata, India}
\and 
Arijit Ghosh \footnotemark[1]
\and 
Gopinath Mishra\footnote{University of Wariwck, UK.  Research supported in part by the Centre for Discrete Mathematics and its Applications (DIMAP), by EPSRC award EP/V01305X/1.}
}
}

\date{}
\maketitle

\thispagestyle{empty} 
\begin{abstract}
\noindent
{Beame et al.\ [ITCS'18 \& TALG'20] introduced and used the {\sc Bipartite Independent Set} ({\sc BIS}) and {\sc Independent Set} ({\sc IS}) oracle access to an unknown, simple, unweighted and undirected graph and solved the edge estimation problem. The introduction of this oracle set forth a series of works in a short time that either solved open questions mentioned by Beame et al.\ or were generalizations of their work as in Dell and Lapinskas [STOC'18 and TOCT'21], Dell, Lapinskas, and Meeks [SODA'20 and SICOMP'22], Bhattacharya et al.\ [ISAAC'19 \& TOCS'21], and Chen et al.\ [SODA'20]. Edge estimation using {\sc BIS} can be done using polylogarithmic queries, while {\sc IS} queries need sub-linear but more than polylogarithmic queries. Chen et al.\ improved Beame et al.'s upper bound result for edge estimation using {\sc IS} and also showed an almost matching lower bound. Beame et al.\ in their introductory work asked a few open questions out of which one was on estimating structures of higher order than edges, like triangles and cliques, using {\sc BIS} queries.

In this work, we almost resolve the 
query complexity of estimating triangles using {\sc BIS} oracle. While doing so, we prove a lower bound for an even stronger query oracle called Edge Emptiness ({\sc EE}) oracle, recently introduced by Assadi, Chakrabarty, and Khanna [ESA'21] to test graph connectivity. 
%Till now, query oracles were used for commensurate jobs -- \bis and \is for edge estimation, \tislong for triangle estimation, \gpislong for hyperedge estimation. 
%Ours is a work that uses a lower order oracle access, like \bis to estimate a higher order structure like triangle.
}

\end{abstract}

%\tableofcontents
\thispagestyle{empty}

\pagenumbering{arabic}
\newpage
\section{Introduction}\label{sec:intro}

\noindent
Motivated by connections to group testing, to emptiness versus counting questions in computational geometry, and to the
complexity of decision versus counting problems, Beame et al.\ introduced the \bislong (shortened as \bis) and \islong (shortened as \is) oracles as a counterpoint to the local queries~\cite{GoldreichR02,Feige06,GoldreichR08}. The \bis query oracle can be seen in the lineage of the query oracles~\cite{Stockmeyer83,Stockmeyer85,RonT16} that go beyond the local queries. The local queries for a graph $G=(V(G),E(G))$ are: (i) {\sc degree} query: given $u \in V(G)$, the oracle reports the degree of $u$ in $V(G)$; (ii) {\sc neighbour} query: given $(u \in V(G))$, the oracle reports the $i$-th neighbor of $u$, if it exists; otherwise, the oracle reports {\sc null}; (iii) {\sc adjacency} query:  given $u, \, v \in V(G)$, the oracle reports whether $\{u,v\} \in E(G)$. There is another related query oracle {\sc RandomEdge} query: the oracle reports an edge uniformly at random when we query.\footnote{ Note that that, in {\sc RandomEdge} query, the probability space is the set of all edges. So, it is not actually a local query.} 

Let us start by looking into the formal definitions of \bis and \is.

\begin{defi}(\bislong)
 Given disjoint subsets $U, U' \subseteq V(G)$, a \bis query answers whether there exists an edge between $U$ and $U'$ in $G$. 
\end{defi}

\begin{defi}(\islong)
 Given a subset $U \subseteq V(G)$, a \is query answers whether there exists an edge between vertices of $U$ in $G$. 
\end{defi}

The introduction of this new type of oracle access to a graph spawned a series of works that either solved open questions~\cite{DellL18-stoc18,DellLM20-soda} mentioned in Beame et al.\ or were generalizations~\cite{DellLM20-soda,Bhatta-abs-1808-00691,abs-1908-04196}. Beame et al.\ used \bis and \is queries to estimate the number of edges in a graph~\cite{BeameHRRS18,DBLP:journals/talg/BeameHRRS20}. One of their striking observations was that \bis queries were more effective than \is queries for estimating edges. This observation also fits in with the fact that \is queries can be simulated in a randomized fashion using polylogarithmic \bis queries~\footnote{Let us consider an \is query with input $U$. Let us partition $U$ into two parts $X$ and $Y$ by putting each vertex in $U$ to $X$ or $Y$ independently and uniformly at random. Then we make a \bis query with inputs $X$ and $Y$, and report $U$ is an independent set if and only if \bis reports that there is no edge with one endpoint in each of $X$ and $Y$. Observe that we will be correct with at least probability $1/2$. We can boost up the probability by repeating the above procedure suitable number of times.}. 
Edge estimation using \bis was also solved in~\cite{DellL18-stoc18} albeit in a higher query complexity than~\cite{BeameHRRS18}. There were later generalizations of the \bis oracle to estimate higher order structures like triangles and hyperedges~\cite{DellLM20-soda,BhattaISAAC,abs-1908-04196}. On the \is front, Beame et al.'s result for edge estimation using \is oracle was improved in~\cite{CLW-soda-2020} with an almost matching lower bound and thereby resolving the query complexity of estimating edges using \is oracle. One can observe the interest generated in these (bipartite) independent set based oracles in a short span of time. The results are summarized in Table~\ref{table:results-bis-is}: a cursory glance would tell us that commensurate higher order queries were needed for estimating higher order structures (\tislong (shortened as \tis) for counting triangles, \gpislong (shortened as \gpis) for counting hyperedges) if polylogarithmic number of queries is the benchmark. We provide the definitions of \tis and \gpis below. 

\begin{defi}[\tislong ({\sc TIS})~\cite{BhattaISAAC}] Given three disjoint subsets $A,B,C$ of the vertex set $V$ of a graph $G(V,E)$, the {\sc TIS} oracle reports whether there exists a triangle having endpoints in $A, \, B$ and $C$.
\end{defi}

\begin{defi}[\gpislong ({\sc CID})~\cite{BishnuGKM018,DellLM20-soda}] Given $d$ pairwise disjoint subsets of vertices $A_1,\ldots,A_d \subseteq U(\cH)$ of a hypergraph $\cH$ ($U(\cH)$ is the vertex set of the hypergraph $\cH$) as input, \gpis{} query oracle answers whether $m(A_1,\ldots,A_d) \neq 0$, where $m(A_1,\ldots,A_d)$ denotes the number of hyperedges in $\cH$ having exactly one vertex in each $A_i$, $\forall i \in \{1,2,\ldots,d\}$. \end{defi}

Notice the use of number of disjoint subsets in the definition of {\sc TIS} and {\sc CID}. That is why, we call {\sc TIS} and {\sc CID} as higher order query oracles than {\sc BIS} and {\sc IS}.

\begin{table}[thb]
\begin{center}
\begin{tabular}{||c || c | c | c | c||} 
\hline \hline

 \multirow{2}{*} {Work} &  \multirow{2}{*}{ Oracle used} & Structure  &       {Upper bound} & Any other  \\
\cline{4-4}
     &             & estimated  &       {Lower bound} & problem solved? \\
     \hline \hline 
                   \cline{1-5}
                   
 \multirow{2}{*}{~\cite{GoldreichR08}} & \multirow{2}{*}{{\sc Local}} & \multirow{2}{*} {edge} & $\tOh\left(\frac{n}{\sqrt{m}}\right)$ & {Approximating}\\
 \cline{4-4}
                                  &      &      & $\Omega\left(\frac{n}{\sqrt{m}}\right)$ & {average distance.} \\
 
 \cline{1-5}
 
 \multirow{2}{*}{~\cite{CLW-soda-2020}} & \multirow{2}{*}{\is} & \multirow{2}{*} {edge} & $\tOh\left(\min \left\{\sqrt{m}, n/\sqrt{m}\right\} \right)$ & -- \\
 \cline{4-4}
                                  &      &      & $\tOm\left(\min \left\{\sqrt{m}, n/\sqrt{m}\right\}\right)$ & -- \\

 \cline{1-5}

 \multirow{2}{*}{~\cite{BeameHRRS18}} & \multirow{2}{*}{\bis} & \multirow{2}{*} {edge} & $ \poly (\log n, 1/\vareps)=\tOh(1)$ & {Edge estimation} \\
 \cline{4-4}
                                  &      &      & -- & {using \is queries.} \\
 
 \cline{1-5}

 \multirow{2}{*}{~\cite{BhattaISAAC}} & \multirow{2}{*}{\tis} & \multirow{2}{*} {triangle} & $\poly (\log n, \Delta, 1/\vareps)^{\dag}$ & -- \\
 \cline{4-4}
                                  &      &      & -- & -- \\
 \cline{1-5}
 {{~\cite{DellLM20-soda},}} & \multirow{2}{*}{\gpis} & \multirow{2}{*} {hyperedge} & $\poly (\log n, 1/\vareps)=\tOh(1)$ &  ~\cite{DellLM20-soda} resolved \\
 \cline{4-4}
                  ~\cite{abs-1908-04196}  $^\ddag$               &      &      & -- &  Q2  in positive. \\
 
 \cline{1-5}
        
 \multirow{2}{*}{~\cite{EdenLRS15}} & \multirow{2}{*}{{\sc Local}} & \multirow{2}{*} {triangle} & $\tOh\left(\frac{n}{T^{1/3}}+\min\left\{m,\frac{m^{3/2}}{T}\right\}\right)$ & -- \\
 \cline{4-4}
                                  &      &      & $\Omega\left(\frac{n}{T^{1/3}}+\min\left\{m,\frac{m^{3/2}}{T}\right\}\right)$ & -- \\     
  \cline{1-5}          
  \multirow{2}{*}{~\cite{DBLP:conf/innovations/AssadiKK19}} & {{\sc Local}+} & \multirow{2}{*} {triangle} & $\tOh\left(\min\left\{m,\frac{m^{3/2}}{T}\right\}\right)$ & Estimated number \\
 \cline{4-4}
                                  &  {\sc  Random Edge}    &      & $\Omega\left(\min\left\{m,\frac{m^{3/2}}{T}\right\}\right)$ & of arbitrary subgraphs. \\   
                                  \hline \hline  
  \cline{1-5}
  
   \multirow{2}{*}{This work} & \multirow{2}{*}{\bis} & \multirow{2}{*} {triangle} & $\tOh\left(\min\left\{\frac{m}{\sqrt{T}},\frac{m^{3/2}}{T}\right\}\right)$ & -- \\
 \cline{4-4}
                                  &      &      & $\tOm\left(\min\left\{{\sqrt{T}},\frac{m^{3/2}}{T}\right\}\right)$ & -- \\     
  \cline{1-5}      
\hline \hline 
\end{tabular}
\end{center}
\caption{The whole gamut of results involving {\sc Local} queries~\cite{G2017}, \bis, \is and its generalizations. $^\dag$ $\Delta$ is the maximum number of triangles on an edge. $^{\ddag}$ Both these results estimate the number of hyperedges in a $d$-uniform hypergraph, where $d$ is treated as a constant. Here $n, \, m$ and $T$ denote the number of vertices, edges and triangles in  a graph $G$, respectively. $\tOh(\cdot)$ and $\tOm(\cdot)$ hide a multiplicative factor of $\poly(\log n, 1/\eps)$ and $1/\poly(\log n, 1/\eps)$, respectively.}
\label{table:results-bis-is}
\end{table} 
%\end{small}

\normalsize
 
\subsection{The open questions suggested by Beame et al.~\cite{BeameHRRS18,DBLP:journals/talg/BeameHRRS20}}

For a work that has spawned many interesting results in such a short span of time, let us focus on the open problems and future research directions mentioned in~\cite{BeameHRRS18,DBLP:journals/talg/BeameHRRS20}.
\begin{description}
 \item[Q1] 
    Can we estimate the number of cliques using polylogarithmic number of \bis queries? 
 
 \item[Q2] 
    Can polylogarithmic number of \bis queries sample an edge uniformly at random? 
    
 \item[Q3] \remove{\comments{Any graph parameter problem many benefit from use of \bis and \is possibly with when used in combination with local queries for graph parameter estimation problems.  Beame \etal showed that {\sc Edeg Estimation} is an example. Here the main question is there any other graph parameter estimation problem other than of {\sc Edge Estimation} problem?.}}
 Can \bis or \is queries possibly be used in combination with local queries for graph parameter estimation problems? 
 \item[Q4] What other oracles, besides subset queries, allow estimating graph parameters with a polylogarithmic number of
queries? 
\end{description} 
 
\medskip 

\paragraph*{Answers to the above questions and a discussion}

Only Q2 has been resolved till now in the positive~\cite{DellLM20-soda} as can be observed from Table~\ref{table:results-bis-is}. At its core, Q1 asks if a query oracle can step up, that is, if it can estimate a structure that is of a higher order than what the oracle was designed for. The framing of Q1 seems that Beame et al.\ expected a polylogarithmic query complexity for estimation of the number of cliques using \bis. Pertinent to these questions, we also want to bring to focus a work~\cite{abs-2006-14015} where the authors mention that it seems to them that estimation of higher order structures will require higher order queries (see the discussion after Proposition 23 of~\cite{abs-2006-14015}). They showed that $\Omega(n^2/\log n)$ \bis queries are required to separate \emph{triangle free} graph instances from graph instances having at least one triangle. This lower bound follows directly from the communication complexity of triangle freeness testing~\cite{DBLP:conf/soda/Bar-YossefKS02}. However, the full complexity of triangle estimation using emptiness queries like \bis remains elusive.  It seems to us that the observations in ~\cite{BeameHRRS18,DBLP:journals/talg/BeameHRRS20} and ~\cite{abs-2006-14015} about the power of \bis in estimating higher order structures stand in contrast. In this backdrop, we have almost resolved
the lower bound for estimating triangles using \bis queries, and our upper bound result show that 
even if \bis can not estimate triangles using polylogarithmic queries but it is still more powerful 
than {\sc Local + Random Edge} queries on graph for estimating triangles (see Table~\ref{table:results-bis-is}). 
 \bis has an inherent asymmetry in its structure in the following sense -- when \bis says that there exists no edge between two disjoint sets, then \bis stands as a witness to the existence of two sets of vertices having no interdependence, while a yes answer implies that there can be any number of edges, varying from one to the product of the cardinality of the two sets, going across the two sets. We feel that this property of \bis gives it its power, but on the other hand, also makes it difficult to prove lower bounds. That is probably the reason why works related to upper bound for \bis and its generalizations exist, whereas works on lower bound were not forthcoming. Though not on \bis, the work of Chen et al.\ using \is queries gave an interesting lower bound construction for \is oracle. Our work goes one step further in being able to prove a lower bound for the \bis oracle which is much stronger than \is oracle. We have almost resolved the open question Q1 by proving almost matching lower and upper bounds involving \bis for estimating triangles.

%Our result even goes further -- if we want to estimate the number of triangles using a polylogarithmic number of queries, then even a stronger query than \bis (named as \eelong (see Definition~\ref{defi:ee})) is hopeless (see Theorem~\ref{theo:bis-lb})! 

%Our results are summarized in the following theorems, which are the main result we prove in this work. 

\remove{Beame et al.\ had asked if polylogarithmic \bis queries would suffice to estimate the number of cliques in a graph.}

\subsection{A stronger oracle than \bis, our main result and its consequences}
\noindent
Now we define \eelong (shortened as \ee) query oracle which is stronger than both \bis and \is. The \eelong query is a form of a \emph{subset query}~\cite{Stockmeyer83,Stockmeyer85,RonT16} where a subset query with a subset $P \subseteq U$ asks whether $P \cap T$ is empty or not, where $T$ is also a subset of the universe $U$. The \eelong query operates with $U$ being the set of all vertex pairs in $G$, $T$ being the set of edges $E$ in $G$, and $P$ being a subset of pairs of vertices of $V$. 
\begin{defi}(\eelong) \label{defi:ee}
 Given a subset $P \subseteq {V(G) \choose 2}$, a \ee query answers whether there exists an $\{u,v\}\in P$ such that $\{u,v\}$ is an edge in $G$. 
\end{defi}
%(Gopi: Focus on the definitions of subset query and \ee. Are they matching?)}
\ee query is recently introduce by Assadi~\etal~\cite{DBLP:conf/esa/AssadiCK21}~\footnote{Assadi~\etal~\cite{DBLP:conf/esa/AssadiCK21} named \ee query as {\sc OR} query in their paper.}. Note that \bis and \is queries can be simulated by an \ee query \footnote{Let us consider a \bis query with inpus $A$ and $B$. Let $P$ be the set of vertex pairs with one vertex from each of $A$ and $B$. We call \ee oracle with input $P$, and  report there is an edge having one vertex in each of $A$ and $B$ if and only if the \ee oracle reports that there exists an $\{u,v\} \in P$ that forms an edge in $G$.  Similarly, we can simulate an \is query with input $U$ by using an \ee query with input $P={U \choose 2}$.}. We prove our lower bound in terms of the {\em stronger}\footnote{For an example to show how \ee query can be powerful than that of \bis or \is, $\Oh(\log \log n)$ \ee queries~\cite{Stockmeyer83} are enough to estimate the number of edges in a graph, as opposed to high query complexity (compared to $\Oh \left(\log \log n \right)$) when we have \bis or \is queries (See Table~\ref{table:results-bis-is}).} 
\ee queries that will directly imply the lower bound in terms of \bis. But we prove matching upper bound in terms of \bis. Our main results are stated below in an informal setting. The formal statements are given in Theorems~\ref{theo:tri-lb} and \ref{theo:ub}.
\begin{theo}[{\bf Main lower bound (informal statement)}]\label{theo:bis-lb}
Let $m, \, n, \, t \in \N$. Any (randomized) algorithm that has \ee query access to a graph $G(V,E)$ with $n$ vertices and $\Theta(m)$ edges, requires $\widetilde{\Omega}\left(\min\left\{{\sqrt{t}}, \frac{m^{3/2}}{t}\right\}\right)$ \ee queries to  decide whether the number of triangles in $G$ is at most $t$ or at least $2t$.
\end{theo}

\begin{theo}[{\bf An upper bound (informal statement)}]\label{theo:bis-ub}
There exists an algorithm, that has \bis query access to a graph $G(V,E)$,  finds a $(1 \pm \eps)$-approximation to the number of triangles in $G$ with high probabilility, and makes  $\tOh\left(\min\left\{\frac{m}{\sqrt{T}},\frac{m^{3/2}}{T}\right\}\right)$  \bis queries in expectaton.  Here $m, \, n, \, T$ denote the number of vertices, edges and triangles in $G$.
\end{theo}

Note that \eelong query is the \emph{strongest} subset query on edges of the graph. Informally speaking, our lower bound states that no subset query on edges can estimate the number of triangles in a graph by using polylogarithmic queries. However, the results of Bhattacharya \etal~\cite{abs-1908-04196} and Dell \etal~\cite{DellLM20-soda} imply that polylogarithmic \tis queries are enough to estimate the number of triangles in the graph. Note that \tis query is also a subset query on triangles in the graph. To complement our lower bound result, we also give an algorithm (see Theorem~\ref{theo:bis-ub}) for estimating the number of triangles in a graph with \bis queries that matches our lower bound.  Here we would also like to mention that the number of \bis queries our algorithm uses is less than that of the number of {local queries~\cite{G2017}} needed to estimate the number of triangles in a graph. {This implies that we are resolving Q3 in positive in the sense that \bis queries are efficient queries for triangle estimation vis-a-vis local queries~\cite{EdenLRS15} coupled with even random edge queries~\cite{DBLP:conf/innovations/AssadiKK19} (see Table~\ref{table:results-bis-is}).} 

\remove{
It has already been established by the works of Beame et al. and both the upper and lower bounds in Chen et al.~\cite{CLW-soda-2020} that \bis is more powerful than \is in estimating the number of edges. Our lower bound result signifies an interesting fact -- coupled with the lower bound result of Chen et al. it establishes a clear separation in the powers of \bis and \is oracles. \complain{(Gopi, elaborate on this.)}
\comments{I think we can drop this paragraph as we do not have an upper bound in terms of \bis.}
}

\remove{
Open questions include using a polylogarithmic number of BIS queries to estimate the number of cliques
in a graph (see~\cite{EdenRS18} for an algorithm using degree, neighbor, and edge existence queries) or to sample
a uniformly random edge (see~\cite{EdenR18} for an algorithm using degree, neighbor and edge existence queries).
In general, any graph estimation problems may benefit from BIS or IS queries, possibly in combination
with standard queries (such as neighbor queries). Finally, it would be interesting to know what other
oracles, besides subset queries, enable estimating graph parameters with a polylogarithmic number of
queries.
}

%\paragraph*{Which ones are solved?} xx

%\paragraph*{What is our plan?} xx

\subsection{Notations}
\noindent
Throughout the paper, the graphs are undirected and simple.  For a graph $G(V,E)$, $V(G)$ and $E(G)$ denote the set of vertices and edges, respectively; $\size{V(G)}=n$, $\size{E(G)}=m$ and the number of triangles is $T$, unless otherwise specified. We use ${V(G) \choose 2}$ to denote the set of vertex pairs in $G$. Note that $E(G) \subseteq {V(G) \choose 2}$. For $P \subseteq {V(G) \choose 2}$, $V(P)$ represents the set of vertices that belong to at least one pair in $P$.  The neighborhood of a vertex $v \in V(G)$ is denoted by $N_{G}(v)$, and $\size{N_G(v)}$ is called the degree of vertex $v$ in $G$.  $\Gamma(\{x,y\})$ denotes the set $N_G(x)\cap N_G(y)$, that is, the set of common neighbors of $x$ and $y$ in $G$. If $e=\{x,y\}\in E(G)$, $\Gamma(e)$ denotes the set of vertices that forms triangles with $e$ as one of their edges.
\remove{Throughout the paper, $n, \, m$ and $T$ denote the number of vertices, the number of edges and the number of triangles, unless otherwise specified.} The induced degree of a vertex $v$ in $Z \subseteq V(G) \setminus \{v\}$ is the cardinality of $N_G(v) \cap Z$. For $X \subseteq V(G)$, the subgraph of $G$ induced by $X$ is denoted by $G[X]$. Note that $E(G[X])=\{\{x,y\}: x, y \in X\}$. For two disjoint sets $A, \, B \subset V(G)$, the bipartite subgraph of $G$ induced by $A$ and $B$ is denoted by $G[A,B]$. Note that $E(G[A,B])$ is the set of edges having one vertex in $A$ and the other vertex in $B$.

Throughout the paper, $\eps \in (0,1)$ is the approximation parameter. When we say $a$ is a $(1 \pm \eps)$-approximation of $b$, then $(1-\eps)b \leq a \leq (1+\eps) b$. Polylogarithmic means $\poly \left(\log n , {1}/{\eps}\right)$. $\tOh(\cdot)$ and $\tOm(\cdot)$ hide a multiplicative factor of $\poly(\log n, 1/\eps)$ and $1/\poly(\log n, 1/\eps)$, respectively. We have avoided floor and ceiling for simplicity of presentation. The constants in this paper are not taken optimally. We have taken them to let the calculation work with clarity. However, those can be changed to other suitable and appropriate constants.

%\subsection{The Roadmap}

\subsection{Paper organization}
\noindent
We start with the technical overview of our lower and upper bounds in Section~\ref{ssec:lb-overview} and Section~\ref{ssec:ub-overview}, respectively. The detailed lower and upper bound proofs are in Section~\ref{sec:lb-clique} and Appendix~\ref{sec:ub-triangle}, respectively. The missing proofs of Section~\ref{sec:lb-clique} are presented in Appendix~\ref{sec:miss-proofs}. In Appendix~\ref{sec:prelim}, we state some useful probability results.

%\section{\bis: the history and open problems}

%\section{Estimating and sampling using \bis}
\section{Technical overview}
In this section, we discuss about the techniques and proof overview, The overview of the lower bound is discussed in Section~\ref{ssec:lb-overview} and that of the upper bound is discussed in Section~\ref{ssec:ub-overview}.

\subsection{Overview for the proof of our lower bound (Theorem~\ref{theo:bis-lb})}
\label{ssec:lb-overview}
\noindent 
Let us consider $m, \, n, \, t$ as in Theorem~\ref{theo:bis-lb}. We prove the desired bound for \bis (stated in Theorem~\ref{theo:bis-lb}) by proving the lower bound is $\tOm\left(\frac{m^{3/2}}{t}\right)$ when $t \geq km \log n$ and $\tOm\left({\sqrt{t}}\right)$ when $t< km \log n$ for \ee query access, where $k$ is a suitably chosen constant. In this Section, we discuss the overview of the proof when $t\geq k m \log n$. The desired lower bound when $t < km \log n$ can be proved by using a reduction from the case when  $t \geq km \log n$. {The main intuition behind the lower bound is to ``hide" a suitably generated vertex set such that  a large number of queries is necessary to detect such a vertex.}

\paragraph*{The idea for the lower bound of  $\tOm\left(\frac{m^{3/2}}{t}\right)$ when $t \geq km \log n$:\\}
\noindent 
We prove by using Yao's method~\cite{DBLP:books/crc/99/0001R99}. There are two distributions $\yesdist$ and $\nodist$ (as described below) from which $G$ is sampled satisfying $\pr\left(G \sim \yesdist \right)=\pr\left( G \sim \nodist \right)=1/2$. Note that, for each $G \sim \yesdist \cup \nodist$, $\size{V(G)}=n=\Theta(\sqrt{m})$~\footnote{Without loss of generality, we assume that $\sqrt{m}$ is an integer. The proof can be extended to any graph having $n \geq \sqrt{m}$ vertices and $m$ edges by adding $n-\sqrt{m}$ isolated vertices.}, and $\size{E(G)}=\Theta(m)$ with a probability of at least $1-o(1)$. But the number of triangles in each $G \sim \nodist$ is at least two factor more than that of the number of triangles in any $G \sim \yesdist$, with a probability of at least $1-o(1)$.
 \begin{figure}
     \centering
     \includegraphics{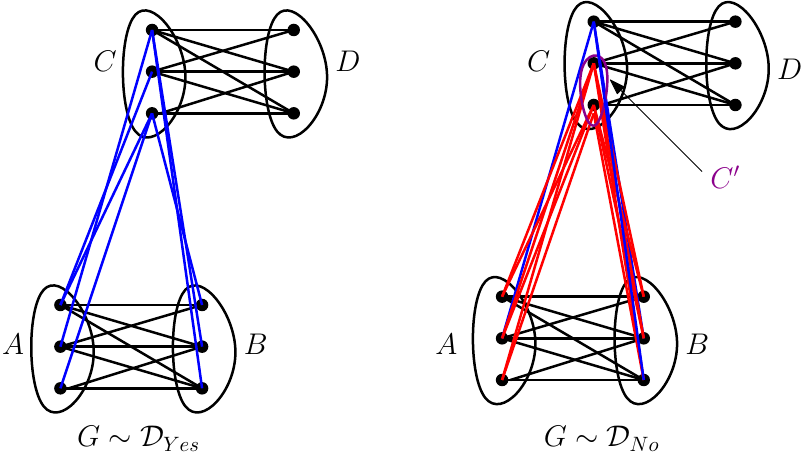}
     \caption{Illustration of $G \sim \yesdist$ and $G \sim \nodist$ when $t=\Omega(m \log n)$.}
     \label{fig:hard}
 \end{figure}
 \begin{description}
 \item[$\yesdist$:] The vertex set $V(G)$ (with $\size{V(G)}=\Theta(\sqrt{m})$) is partitioned into four parts $A, \, B, \, C, \, D$ uniformly at random. Vertex set $A$ forms a biclique with vertex set $B$ and vertex set $C$ forms a biclique with vertex set $D$. Then {every} vertex pair $\{x,y\}$, with $x \in A\cup B$ and $y \in C$, is added as an edge to graph $G$ with probability $ \Theta\left(\sqrt{\frac{t}{m^{3/2}}}\right)$; 
 
 \item[$\nodist$:] The vertex set $V(G)$ (with $\size{V(G)}=\Theta(\sqrt{m})$) is partitioned into four parts $A, \, B, \, C, \, D$ uniformly at random. Vertex set $A$ forms a biclique with vertex set $B$ and vertex set $C$ forms a biclique with vertex set $D$. Then {every} vertex pair $\{x, y\}$, with $x \in A \cup B$ and $y \in C$, is added as an edge to graph $G$ with probability $ \Theta\left(\sqrt{\frac{t}{m^{3/2}}}\right)$. Then each vertex of $C$ is sampled with probability $\Theta\left(\frac{t}{m^{3/2}}\right)$. Let $C'$ be the sampled set. Each vertex of $C'$ is connected to every vertex of $A \cup B$ with an edge;
 \end{description}
See Figure~\ref{fig:hard} for an illustration of the above construction. The constants, including $k$, in the order notations above are suitably set to have the following:
 \begin{description}
 \item[When $G \sim \yesdist$:]  The number of triangles in each graph is at most $t$, with a probability of at least $1-o(1)$;
 \item[When $G \sim \nodist$:] $\size{C'}=\Theta\left(\frac{t}{m}\right)$ with a probability of at least $1-o(1)$. Hence, the number of triangles in each $G \sim \nodist$ is at least $2t$, with a probability of at least $1-o(1)$.
\end{description}

Now, consider a particular \ee query with input $P \subseteq {V(G) \choose 2}
$. Here, we divide the discussion into two parts, based on $\size{P} 
\leq \tau$ and $\size{P}>\tau$, where $\tau=\Theta(\log ^2 n)$ {is a threshold. If we query with the number of vertex pairs more than the threshold, chances are more we will not be able to distinguish between $G \sim \yesdist$ and $G \sim \nodist$.} When $
\size{P}>\tau$, we can show that there exists a vertex pair $\{x,y\}\in 
P$ such that $\{x,y\}$ is an edge in $G$, with a probability of at least 
$1-o(1)$, irrespective of whether $G \sim \yesdist$ or $G \sim \nodist$. 
Intuitively, this is because the number of vertices and edges in $G$ are 
$\Theta(\sqrt{m})$ and $\Theta(m)$, respectively. So, \ee queries with 
input $P \subseteq {V(G) \choose 2}$ such that $\size{P}>\tau$ will not 
be useful to distinguish whether $G \sim \yesdist$ or $G \sim \nodist$.
\remove{So, the \ee query with input $P \subseteq {V(G) \choose 2}
$ are not \emph{useful} to distinguish whether $G\sim \yesdist$ or $G\sim \nodist$, if $ \size{P} > \tau$. (Gopi: Are these two sentences not saying the same thing? Can you use just one sentence?)}

 We prove the desired lower bound by proving $
\tOm\left(\frac{m^{3/2}}{t}\right)$ \ee queries are necessary to 
decide whether $G\sim \yesdist$ or $G\sim \nodist$ with a probability of 
at least $2/3$. Note that $C'=\emptyset$ when $G \sim \yesdist$. So, the number of \ee queries {needed} to decide whether $G\sim 
\yesdist$ or $G\sim \nodist$, is at least the number of \ee queries {needed} to 
\emph{touch} at least one vertex of $C'$ when $G \sim \nodist$. Here, by touching at least a vertex of $C'$, we mean $V(P) \cap C' \neq \emptyset$. As we have argued that only \ee query with input $P \subseteq {V(G) \choose 2}$ with $\size{P} \leq \tau$ can be useful, the probability that we touch a vertex in $C'$ with such a query is at most $ p=\Oh\left(\frac{C'}{\sqrt{m}}\cdot \tau\right)= \Oh\left(\frac{t \log ^2 n}{m^{3/2}}\right)$. Hence the number of \ee queries to touch at least a vertex of $C'$, is at least $1/p$, that is, $\tOm\left(\frac{m^{3/2}}{t}\right)$. 

To let the the above discussion work, when $G \sim \nodist$, $\size{C'}=\Theta\left(\frac{t}{m}\right)$ must be at least $\Omega(\log n)$. But $\size{C'}=\Theta\left(\frac{t}{m}\right)$ with a probability of at least $1-o(1)$. Because of this, we take $t\geq km \log n$ in the above discussion. The formal statement of the lower bound, when $t \geq km \log n$, is given in Lemma~\ref{theo:clq-count} in Section~\ref{sec:lb-clique}. What we have discussed here is just an overview, the formal proof of Lemma~\ref{theo:clq-count} is much more invloved and delicate, which is presented in Section~\ref{sec:pf-high}.

As we have already mentioned, the desired lower bound of $\tOm\left(\sqrt{t}\right)$ when $t < km\log n$ can be proved by a reduction from the case when $t \geq km\log n$. The formal statement of the lower bound, when $t < km \log n$, is given in Lemma~\ref{theo:clq-count-low} in Section~\ref{sec:lb-clique}, and the proof is presented in Section~\ref{sec:pf-low}.

%\subsection{Overview for our upper bound (Theorem~\ref{theo:bis-ub})} 

\subsection{Overview for our upper bound (Theorem~\ref{theo:bis-ub})} 
\label{ssec:ub-overview}
\noindent
We establish the upper bound claimed in  Theorem~\ref{theo:bis-ub} by giving two algorithms that report a $(1\pm \eps)$-approximation to the number of triangles in the graph:
\begin{itemize}
\item[(i)] $\alghigh(G,\eps)$ that makes $\tOh\left(\frac{m^{3/2}}{T}\right)$ \bis queries;
\item[(ii)] $\alglow(G,\eps)$ that makes $\tOh\left(\frac{m+T}{\sqrt{T}}\right)$ \bis queries.
\end{itemize}
Informally speaking, our final algorithm $\test$ calls $\alghigh (G,\eps)$ and $\alglow(G,\eps)$ when $T=\Omega(m)$ and $T=\Oh(m)$, respectively. Observe that, if $\test$ knows $T$ within a constant factor, then it can decide which one to use among  $\alghigh (G,\eps)$ and $\alglow(G,\eps)$. If $\test$ does not know $T$ within a constant factor, then it starts from a guess $L={{n \choose 3}}/{2}$ and makes a geometric search on $L$ until the output of $\test$ is consistent with $L$. Depending on whether $L =\Omega(m)$ or $L=\Oh(m)$, $\test$ decides which one among $\alghigh (G,\eps)$ and $\alglow(G,\eps)$ to call. This guessing technique is very standard by now in property testing literature~\cite{GoldreichR08,EdenLRS15,EdenRS18, DBLP:conf/innovations/AssadiKK19}. Another point to note is that we do not know $m$. However, we can estimate $m$ by using $\poly (\log n)$ \bis queries (see Table~\ref{table:results-bis-is}). An estimate of $m$ will perfectly work for us in this case.

\paragraph*{Algorithm $\alghigh (G,\eps)$} 

Algorithm \alghigh is inspired by the triangle estimation algorithm of Assadi \etal ~\cite{DBLP:conf/innovations/AssadiKK19}, where we have {\sc Adjacency}, {\sc Degree}, {\sc Random neighbor} and {\sc Random Edge} queries. Please see Appendix~\ref{sec:alg-high} for formal definitions of these queries. Note that the algorithm by Assadi \etal~can be \emph{suitably} modified even if we have approximate versions of {\sc Degree}, {\sc Random neighbor} and {\sc Random Edge} queries. Also refer Appendix~\ref{sec:alg-high} for formal definitions of approximate version of the above queries. By Corollary~\ref{coro:bis}, $\tOh(1)$ \bis queries are enough to simulate the approximate versions of {\sc Degree} and {\sc Random neighbor}, with a probability of at least $1-o(1)$. By Proposition~\ref{pro:bis-ub3}, approximate version of {\sc Random Edge} queries can also be simulated by $\tOh(1)$ \bis queries, with a probability of at least $1-o(1)$. Putting everything together, we get $\alghigh (G,\eps)$ for triangle estimation that makes $\tOh(1)$ \bis queries. The formal statement of the corresponding triangle estimation result  is given in Lemma~\ref{lem:tri-count}, and algorithm $\alghigh (G,\eps)$ is described in Appendix~\ref{sec:alg-high}.

\paragraph*{Algorithm $\alglow (G,\eps)$} 

This algorithm is inspired by the two pass streaming algorithm for triangle estimation by McGregor \etal~\cite{DBLP:conf/pods/McGregorVV16}. Basically, we show that the steps of McGregor \etal's algorithm
can be executed by using $\tOh\left(\frac{m+T}{\sqrt{T}}\right)$ \bis queries. To do so, we have used the fact that, given any $X \subseteq V(G)$, all the edges of the subgraph induced by $X$ can be enumerated by using $\tOh\left(\size{E(G[X])}\right)$ \bis queries (see Proposition~\ref{pro:bis-ub1} for the formal statement). The formal statement of the corresponding triangle estimation result is given in Lemma~\ref{lem:tri-count-low}, and algorithm $\alglow (G,\eps)$ is described in Appendix~\ref{sec:alg-low} along with its correctness proof and query complexity analysis.

%The above description is based on as if we know $m$ and $T$, but actually, we do not.  

\section{Lower bound for estimating triangles using \eelong queries}
\label{sec:lb-clique}
\noindent
In this Section, we prove the main lower bound result as sketched in Theorem~\ref{theo:bis-lb}; the formal theorem statement is stated below. As mentioned earlier, the lower bound proofs will be for the stronger query oracle \ee. This will imply the lower bound for \bis. 

\begin{theo}[{\bf Main lower bound result}]\label{theo:tri-lb}
Let $m, \, n, \, t \in \N$ be such that $1 \leq t \leq \frac{m^{3/2}}{2}$. Any 
(randomized) algorithm that has \ee oracle access to a graph $G(V,E)$ must make 
$\widetilde{\Omega}\left(\min
\left\{{\sqrt{t}},\frac{m^{3/2}}{t}\right\}\right)$ \ee queries to decide whether the number of triangles in $G$ is at most $t$ or at least $2t$ with a probability of at least  $2/3$, where $G$ has $ n \geq 4\sqrt{m}$ vertices and $\Theta\left(m \right)$ edges.% promised to have $\Oh(t)$ triangles.
\end{theo}

We prove the above theorem by proving  Lemmas~\ref{theo:clq-count} and~\ref{theo:clq-count-low}, as stated below. Note that Lemmas~\ref{theo:clq-count} and~\ref{theo:clq-count-low} talk about the desired lower bound when the number of triangles in the graph is \emph{large} ($\Omega(m \log n)$) and \emph{small} ($\Oh(m \log n)$), respectively.  
\begin{lem}[{\bf Lower bound when there are \emph{large} number of triangles}]\label{theo:clq-count}
Let $m, \, n, \, t \in \N$ be such that $t \geq \frac{m \log n}{8}$. Any 
(randomized) algorithm that has \ee oracle access to a graph $G(V,E)$ must make 
$\widetilde{\Omega}\left(\frac{m^{3/2}}{t}\right)$ \ee queries to decide whether the number of triangles in $G$ is at most $t$ or at least $2t$ with a probability of at least  $2/3$, where $G$ has $ n \geq 4\sqrt{m}$ vertices, $\Theta\left(m \right)$  edges.% promised to have $\Oh(t)$ triangles.
\end{lem}
%We can make the above lemma work for if the number of vertices $n$ in $G$ is any integer more than $4\sqrt{m}$ by adding $n-4\sqrt{m}$ many isolated vertices.
\begin{lem}[{\bf Lower bound when there are \emph{small} number of triangles}]\label{theo:clq-count-low}
{
Let $m, \, n, \, t \in \N$ be such that $ t <\frac{m \log n}{8} $. Any 
(randomized) algorithm that has \ee oracle access to a graph $G(V,E)$ must make 
$\widetilde{\Omega}\left({\sqrt{t}}\right)$  \ee queries to decide whether the number of triangles in $G$ is at most $t$ or at least $2t$ with a probability of at least $2/3$, where $G$ has $ n \geq 4\sqrt{m}$ vertices and $\Theta\left(m \right)$ edges.}
\end{lem}
%\color{Black}
We first show Lemma~\ref{theo:clq-count} in Section~\ref{sec:pf-high}, and then Lemma~\ref{theo:clq-count-low} in Section~\ref{sec:pf-low}. Note that the proof of Lemma~\ref{theo:clq-count-low} will use Lemma~\ref{theo:clq-count}.

\subsection{Proof of Lemma~\ref{theo:clq-count}}\label{sec:pf-high}
\noindent
 Without loss of generality, assume that $\sqrt{m}$ is an integer. We prove for the case when $n=4\sqrt{m}$. But, we can make the proof work for any $n \geq 4\sqrt{m}$ by adding $n-4\sqrt{m}$ isolated vertices.
 Note that $ t \geq \frac{m \log n}{8}$ here. We further assume that $t \leq \frac{m^{{3}/{2}}}{128}$, and {$m =\Omega(\log ^{2} n)$}. Otherwise, the stated lower bound of $\widetilde{\Omega}\left(\frac{m^{3/2}}{t}\right)$ trivially follows as $\widetilde{\Omega}(\cdot)$ hides 
a multiplicative factor of $\frac{1}{\poly (\log n)}$.

We use Yao's min-max principle to prove the lower bound. To do so, we consider two distributions $\yesdist$ and $\nodist$ on 
graphs where 
\begin{itemize}
\item Any graph $G \sim \yesdist \cup \nodist$ has $ 4\sqrt{m}$ vertices;
\item Any graph $G \sim \yesdist \cup \nodist$ has $\Theta(m)$ edges with a probability of at least $1-o(1)$;
\item The number of triangles in any graph $G \sim \yesdist$ is at most $t$ with a probability of at least $1-o(1)$, and any graph $G \sim \nodist$ has at least $2t$  triangles with a probability of at least $1-o(1)$. 
\end{itemize}
Note that, if we can show that any deterministic algorithm that distinguishes graphs from $\yesdist$ and $\nodist$, with a probability of at least $2/3$, must make $\widetilde{\Omega}\left(\frac{m^{3/2}}{t}\right)$ \ee queries, then we are done with the proof of Lemma~\ref{theo:clq-count}.
\subsubsection{The (hard) distribution for the input, its properties, and the proof set up}\label{sec:hard}
\begin{description}
\item[$\yesdist$:] A graph $G \sim \yesdist$ is sampled as follows:
\begin{itemize}
\item Partition the vertex set $V(G)$ into $4$ parts $A, \, B, \, C, \, D$, by initializing $A, \, B, \, C, \, D$ as empty sets, and then putting each vertex in $V(G)$ into one of the parts uniformly at random and independent of other vertices;
\item  Connect each vertex of $A$ with every vertex of $B$ with an edge {to form a biclique}. Also, connect 
each vertex of $C$ with every vertex of $D$ with an edge {to form another biclique};

\item For {every $\{x,y\}$ where} $x \in A \cup B$ and $y \in C$, add edge $\{x,y\}$ to $G$ with probability $\sqrt{\frac{t}{16m^{3/2}}}$.
\end{itemize} 
\item[$\nodist:$] A graph $G \sim \nodist$ is sampled as follows:
\begin{itemize}
\item Partition the vertex set $V(G)$ into $4$ parts $A, \, B, \, C, \, D$, by initializing $A, \, B, \, C, \, D$ as empty sets, and then putting each vertex in $V(G)$ into one of the partitions uniformly at random and independent of other vertices;
\item  Connect each vertex of $A$ with every vertex of $B$ with an edge {to form a biclique}. Also, connect 
each vertex of $C$ with every vertex of $D$ with an edge {to form another biclique};

\item For {every $\{x,y\}$ where} $x \in A \cup B$ and $y \in C$, add edge $\{x,y\}$ to $G$ with probability $\sqrt{\frac{t}{16m^{3/2}}}$.
\item Select $C' \subseteq C$ by putting each $x \in C$ into $C'$ with a probability of at least $\frac{32t}{m^{3/2}}$, independently, and then, add each edge in $\{{x,y}:x \in A\cup B, y \in C'\}$ to $G$. 
\end{itemize} 
\end{description}

The following observation establishes the number of vertices, edges, and the number of triangles in the graphs that can be sampled from $\yesdist \cup \nodist$. The proof uses \emph{large deviation inequalities} (see Lemma~\ref{lem:cher_bound1} and ~\ref{lem:depend:high_prob} in Appendix~\ref{sec:prelim}), and is presented in Appendix~\ref{sec:miss-proofs}.

\begin{obs}[{\bf Properties of the graph $G \sim \yesdist \cup \nodist$}]\label{obs:edge}
\begin{enumerate}
\item[(i)] For $G\sim \yesdist \cup \nodist$, the number of vertices in $G$ is $4\sqrt{m}$. Also, $\frac{\sqrt{m}}{2} \leq \size{A}, \size{B}, \size{C},\size{D} \leq 2\sqrt{m}$ holds with a probability of at least $1-o(1)$, and the number of edges in $G$ is $\Theta(m)$ with a probability of at least $1-o(1)$;
\item[(ii)] If $G\sim \yesdist$, then there are at most  $t$ triangles in $G$ with a probability of at least $1-o(1)$,
\item[(iii)] If $G\sim \nodist$, $ \frac{8t}{m}\leq \size{C'}\leq \frac{64t}{m}$ with a probability of at least $1-o(1)$, and there are at least $2t$ triangles in $G$ with a probability of at least $1-o(1)$.
\end{enumerate}
\end{obs}

The following remark is regarding the connection between graphs in $\yesdist$ and that in $\nodist$. This will be used later in our proof, particularly in the proof of Claim~\ref{clm:inter2}.
\begin{rem}[{\bf A graph $G' \sim \nodist$ can be generated from a graph  $G \sim \yesdist$}]\label{rem:yes-no}
Let us first generate a graph $G \sim \yesdist$. Select $C' \subseteq C$ by putting each $x \in C$ into $C'$ with a probability of at least $\frac{32t}{m^{3/2}}$, and then, add each edge in $\{{x,y}:x \in A\cup B, y \in C'\}$ to $G$ to generate $G'$, then (the resulting graph) $G'\sim \nodist$.
\end{rem}
The following observation says that a $\{x,y\}\in {V(G)\choose 2}$ (with some condition) forms an edge with a probability of at least a constant. It will be used while we prove Claim~\ref{clm:inter1}.

\begin{obs}[{\bf Any vertex {pair} $\{x,y\}$ is an edge in $G$ with constant probability}]\label{obs:prob-edge}
Let $\{x,y\}\in {V(G)\choose 2}$, and we are in the process of generating $G\sim \yesdist \cup \nodist$. Let at most one of $x$ and $y$ has been put into one of the parts out of $A,B,C$ and $D$. Then $\{x,y\}$ is an edge in $G$ with probability at least $\frac{1}{4}$.
\end{obs}
The above observation follows from the description of $G \sim \yesdist \cup \nodist$ -- each vertex in $V(G)$ is put into one of the parts out of $A, \, B, \, C, \, D$ uniformly at random, each vertex of $A$ is connected with every vertex in $B$, and each vertex of $C$ is connected with every vertex in $D$.
%\subsection{Proving the lower bound for the hard distributions}

In order to prove Lemma~\ref{theo:clq-count}, by contradiction, assume that there is a randomized algorithm that makes $q=o\left(\frac{m^{3/2}}{t}\frac{1}{\log^2  n}\right)$ \ee queries and decides whether the number of triangles in the input graph is at most $t$ or 
at least $2t$, with a probability of at least $2/3$. Then there exists a deterministic algorithm $\alg$ that makes $q$  \ee queries and decides the following (when the input graph $G\sim \yesdist \cup \nodist$ be such that both $G \sim \yesdist$ and $G \sim \nodist$ holds with probability $1/2$) --
$$\pr_{G \sim \nodist}(\mbox{\alg (G) reports {\sc NO}})-\pr_{G \sim \yesdist}(\mbox{\alg (G) reports {\sc NO}}) \geq \frac{1}{3}-o(1).$$
(Here $\pr_{G \sim \nodist}(\cE)$ and $\pr_{G \sim \yesdist}(\cE)$ denote the probability of the event $\cE$ under the conditional space $G \sim \nodist$ and $G \sim \yesdist$, respectively.)
%\complain{(Gopi: does the notation $\pr_{G \sim \nodist}$ need any explanation?)}
Hence, we will be done with the proof of Lemma~\ref{theo:clq-count} by showing the following lemma.
\begin{lem}[{\bf Lower bound on the number of \ee queries when $G \sim \yesdist \cup \nodist$}]\label{lem:decide}
Let the unknown graph $G$ be such that $G \sim \yesdist$ and $G \sim \nodist$ hold with equal probabilities. Consider any deterministic algorithm $\alg$ that has \ee access to $G$, and makes $q=o\left(\frac{m^{3/2}}{t}\frac{1}{\log^2  n}\right)$  \ee queries to $G$. Then
$$\pr_{G \sim \nodist}(\mbox{\alg (G) reports {\sc NO}})-\pr_{G \sim \yesdist}(\mbox{\alg (G) reports {\sc NO}}) \leq o(1) .$$
\end{lem}

Next, we define an augmented \ee oracle (\staree oracle). \staree is tailor-made for the graphs coming from $\yesdist \cup \nodist$. Moreover, it is \emph{stronger} than \ee, that is, any \ee query can be simulated by a \staree query. We will prove the claimed lower bound in Lemma~\ref{lem:decide} when we have access to \staree oracle. Note that this will imply Lemma~\ref{lem:decide}.

Before getting into the formal description of \staree oracle, note that the algorithm (with \staree access) maintains a four tuple data structure initialized with $\emptyset$. With each query to \staree oracle, the oracle updates the data structure and returns the updated data structure to the algorithm. Note that the updated data structure is a function of all previously made \staree queries, and it is enough to answer corresponding \ee queries.
 %\complain{(Gopi: Should we not add here one/two sentences about what is \staree before its description in the next subsection?)}

\subsubsection{Augmented \eelong oracle (\staree)}

Before describing the \staree query oracle and its interplay with the algorithm, first we present the data structure $(E_Q,V\left(E_Q \right), e,\ell_v)$ that the algorithm maintains with the help of \staree oracle. The data structure keeps track of the following information. 

\paragraph*{Information maintained by $(E_Q,V\left(E_Q \right), e,\ell_v)$:}

\begin{itemize}
\item $E_Q$ is a subset of ${V(G) \choose 2}$ that have been seen by the algorithm till now, $V\left(E_Q \right)$ is the set of vertices present in any vertex pair in $E_Q$. 
\item $e: {V\left(E_Q \right) \choose 2} \rightarrow \{0,1\}$ such that $e(\{x,y\})=1$ means the algorithm knows that $\{x,y\}$ is an edge in $G$, $e(x,y)=0$ means that $\{x,y\}$ is not an edge in $G$.
\item $\ell_v: V\left(E_Q \right) \rightarrow \{ A, \, B, \, C, \, C', \, D\}$, where %$\ell_v(x)=A$ if $x \in A$, $\ell_v(x)=B$ if $x \in B$, $\ell_v(x)=C$ if $x \in C$, $\ell_v(x)=C'$ if $x \in C'$, and $\ell_v(x)=D$ if $x \in D$.
{
\[ \ell_v(x)=  \left\{
\begin{array}{ll}
       A, & x \in A \\
      B, & x \in B \\
      C, & x \in C \setminus C' \\
      C', & x \in C' \\
      D, & x \in D \\
\end{array} 
\right. \]
}
\end{itemize} 
Intuitively speaking, unless the algorithm knows about the presence of some vertex in $C'$, it cannot distinguish whether the unknown graph $G \sim \yesdist$ or $G \sim \nodist$. So, we define the notion of good and bad vertices, along with good and bad data structures. This notion will be used later in our proof.
\begin{defi}[{\bf Bad vertex}]\label{defi:bad-dist}
 A vertex $x \in V(E_Q)$ is said to be a \emph{bad} vertex if $\ell_v(x)=C'$. $(E_Q,V\left(E_Q \right), e,\ell_v)$ is said to be  \emph{good} if there does not exist any {bad vertex in $V\left(E_Q \right)$.} \remove{$x \in V\left(E_Q \right)$ which  is a bad vertex.}
\end{defi}

\subsubsection*{\staree oracle and its interplay with the algorithm:}
\noindent
 The algorithm initializes the data structure $(E_Q,V\left(E_Q \right), e,\ell_v)$ with $E_Q =\emptyset$, $V\left(E_Q \right)=\emptyset$. So, $e$ and $\ell_v$  are initialized with trivial functions with domain $\emptyset$. At the beginning of each round, the algorithm queries  the \staree oracle with a subset $P \subseteq {V(G) \choose 2}$  deterministically. Note that the choice of $P$  depends on the current status of the data structure. 
Now, we explain how \staree oracle responds to the query and how the data structure is updated. 
\begin{enumerate}
\item[(1)] If $\size{P} \leq \tau =25 \log^2 n$,  the oracle sets $E_Q \leftarrow E_Q \cup P$, and changes $V\left(E_Q \right)$ accordingly. The oracle also sets the function $e$ and $\ell_v$ as per their definitions, and then it sends the updated data structure to the algorithm.
\remove{\item \staree oracle first checks for a good edge $\{x,y\}$ with $x\in X$ and $y\in Y$. If such an edge $\{x,y\}$  is found, then update $e(\{x,y\})$ as $1$ and $\ell_e(\{x,y\})=\mbox{good}$. For any $x',y' \in K$ such that $e(\{x',y'\})$ is undefined till now, set $e(\{x',y'\})=\perp$ and $\ell_e(\{x,y\})=\mbox{good}$;
\item[(3)] If \staree oracle does not find any good edge $\{x,y\}$ with $x\in X$ and $y\in Y$, then it checks for a bad edge $\{x,y\}$ with $x\in X$ and $y\in Y$. If such an edge $\{x,y\}$ is 
 is found, then update $e(\{x,y\})$ as $1$ and $\ell_e(\{x,y\})=\mbox{bad}$;

% Otherwise, that is, if there is no edge $\{x,y\}$ with $x\in X$ and $y \in Y$, then \staree sets  $e(\{x,y\})=0$ and $\ell_e(\{x,y\})=\mbox{good}$ for every $x\in X$ and $y ]in Y$.
}

\item[(2)]  Otherwise (if $\size{P}>\tau$), the oracle finds a random subset $P' \subseteq P$ such that $\size{P'}=\tau.$ The oracle checks if there is a pair $\{u,v\}\in P'$ such that $\{u,v\}$ is an edge. If yes, then the oracle responds as in (1) with $P$ being replaced by $P'$. If no, the oracle sends the data structure corresponding to the entire graph along with a {\sc Failure} signal~\footnote{We later argue that  {\sc Failure} signal is sent with a very low probability.}.

\end{enumerate}

 Owing to the way \staree oracle updates the data structure after each \staree query, we can make some assumptions on the inputs to the \staree oracle, as described in Remark~\ref{rem:assumption}. It will actually be useful when we prove Claim~\ref{clm:inter1}.
% we can make some assumptions on the inputs to \staree oracle, as described in the
\begin{rem}[{\bf Some assumptions on the \staree query}]\label{rem:assumption}
Let $(E_Q,V\left(E_Q \right), e,\ell_v)$ be the data structure just before the algorithm makes \ee query with input $P$, and let $\left(E_Q',V\left(E_Q' \right), e',\ell_v'\right)$ be the data structure updated by \staree oracle after the algorithm makes \staree query with input $P$. Without loss of generality, we assume that
\begin{itemize}
\item[(i)] $P$ is disjoint from ${V\left(E_Q \right) \choose 2}$. It is because \staree maintains whether $\{x,y\}$ is an edge or not for each $\{x,y\} \in {V\left(E_Q \right) \choose 2}$;
\item[(ii)] When $x, z \in V(E_Q)$, there does not exist $\{x,y\}$ and $\{y,z\}$ in $P$. It is because the oracle updates the data structure in the same way in each of the following three cases when $x,z \in V(E_Q)$ -- (i) $\{x,y\}$ and $\{y,z\}$ are in $P$, (ii) $\{x,y\} \in P$ and $\{y,z\} \notin P$, and (iii) $\{x,y\} \notin P$ and $\{y,z\} \in P$. By the description of \staree oracle and its interplay with the algorithm, in all the three cases, the updated data structure $(E_Q',V\left(E_Q' \right), e',\ell_v')$ contains labels of all the three vertices $x,y,z$  along with the information whether $\{x,y\}$ and $\{y,z\}$ form edges in $G$ or not. So, instead of having both $\{x,y\}$ and $\{y,z\}$ in $P$ with $x,z \in V(E_Q)$, it is \emph{equivalent} to have exactly one among  $\{x,y\}$ and $\{y,z\}$ in $P$. 
%then the updated .% Consider the case  when $\{x,y\} \in P$ o $\{y,z\}$ is in $P$.   the updated data structure $(E_Q',V\left(E_Q' \right), e',\ell_v')$ contains labels of all three vertices $x,y,z$  along with the information whether $\{x,y\}$ and $\{y,z\}$ form edges in $G$ or not. (Gopi: The meaning is not clear from this sentence.)

\end{itemize} 
\end{rem}
 \remove{the oracle  After that \staree oracle first checks for a good edge $\{x,y\}$ 
with $x\in X$ and $y\in Y$. 
If No updates $K$ by $K \cup X \cup Y$.
. Also, the oracle finds the status of each $\{x,y\}$ with $x\in X$ and $y \in Y$ whether $\{x,y\}$ is an edge or not and update $e(\{x,y\})$ accordingly. For any two distinct vertices $x$ and $y$ in $X\setminus K$ ($Y\setminus K$), the oracle sets $e(\{x,y\})=\perp$. Also, for each $x\in (X \cup Y) \setminus K$ and $y \in K \setminus (X \cup Y)$, the oracle sets 
$e(\{x,y\})=\perp$. Finally, the oracle updates $\ell_e$ according to the updated $\ell_v$ and $e$. Then the new data structure is returned to the algorithm.}

In the following observation, we formally show that \staree oracle is stronger than that of \ee. Then we prove Lemma~\ref{lem:decide-star} that says that $\Omega\left(\frac{m^{3/2}}{t}\frac{1}{\log ^2 n}\right)$ \staree queries are necessary to distinguish between $G \sim \yesdist$ and $G \sim\nodist$. Note that Lemma~\ref{lem:decide-star} will imply Lemma~\ref{lem:decide}.

 %So, we will be done with the proof of  by proving Lemma~\ref{lem:decide-star}, which says 
 
\begin{obs}[{\bf \staree is stronger than \ee}]\label{obs:lem_eqv_aug}
Let $G \sim  \yesdist \cup \nodist$. Each \ee query to $G$ can be simulated by using an \staree query to $G$. 
\end{obs}
{\begin{proof}
Let us consider an \ee query with input $P \subseteq {V(G) \choose 2}$. We make a \staree query with the same input $P$, and answer the \ee query as follows depending on whether $\size{P}\leq \tau$ or $\size{P}>\tau$.
\begin{description}
\item[$\size{P}\leq \tau$:] The \staree oracle updates the data structure and let $(E'_Q,V\left(E'_Q \right), e',\ell'_v)$ be the updated data structure. It contains the the information about each $\{x,y\} \in P$ whether it forms an edge in $G$ or not. So, from $(E'_Q,V\left(E'_Q \right), e',\ell'_v)$, the \ee query with input $P$  can be answered as follows: there exists an edge $\{x,y\}\in P$ with $\{x,y\} \in E(G)$ if and only if $e'(\{x,y\})=1$.

\item[$\size{P}> \tau$:] In this case, the \staree oracle finds a random subset $P' \subseteq P$ such that $\size{P'}=\tau$. It checks if there is an $\{x,y\}\in P'$ such that $\{x,y\}$ is an edge. If yes, the updated data structure contains the the information about each $\{x,y\} \in P'$ whether it forms an edge. In this case, we can report that there exists an $\{x,y\} \in P$ such that $\{x,y\}$ is an edge in $G$. If there is no $\{x,y\}\in P'$ such that $\{x,y\}$ is an edge, then (by the description of \ee oracle and its interplay with the algorithm) the \staree oracle sends the data structure corresponding to the entire graph. Obviously, we can report whether there exists an $\{x,y\}\in P$ such that $\{x,y\} \in E(G)$ or not.  

\end{description}

Hence, in any case, we can report the answer to \ee query with input $P$.
\end{proof} }
We are left with proving the following technical lemma. As noted earlier, this will imply Lemma~\ref{lem:decide}.

\begin{lem}[{\bf Lower bound on the number of \staree queries when $G \sim \yesdist \cup \nodist$}]\label{lem:decide-star}
Let the unknown graph $G$ be such that $G \sim \yesdist$ and $G \sim \nodist$ hold with equal probabilities. Consider any deterministic algorithm \staralg that has \staree access to $G$, and makes $q=o\left(\frac{m^{3/2}}{t}\frac{1}{\log^2  n}\right)$  \staree queries to $G$. Then
$$\pr_{G \sim \nodist}(\mbox{\staralg(G) reports {\sc NO}})-\pr_{G \sim \yesdist}(\mbox{\staralg(G) reports {\sc NO}}) \leq o(1).$$

\end{lem}
% Let by contradiction assume that there exists an algorithm \staralg that has \staree query access to a graph $G \sim \yesdist \cup \nodist$ and distinguishes whether $G\sim \yesdist$ or $G \sim \nodist$ by making $q=o\left(\frac{m^{3/2}}{t}\frac{1}{\log ^{2} n}\right)$ queries. 

\subsubsection{Proof of Lemma~\ref{lem:decide-star}}
\noindent
For clarity of explanation, we first describe \staralg{} as a decision tree. Then we will prove Lemma~\ref{lem:decide-star}.
 \subsubsection*{Decision tree view of \staralg:} 
 %Let us view \staralg as a decision tree $\cT$ of depth $q$.
 \begin{itemize}
     \item Each internal node of $\cT$ is labeled with a nonempty subset ${V(G) \choose 2}$ and each leaf node is labeled with {\sc YES} or {\sc NO};
     \item Each edge in the tree is labeled with a data structure $(E_Q,V\left(E_Q \right), e,\ell_v)$;
     \item  The algorithm starts the execution from the root node $r$ by setting $r$ as the current node. Note that for the root node  $r$, $E_Q=V\left(E_Q \right)=\emptyset$  and $e$ and $\ell_v$ are the trivial functions. As the algorithm \staralg is deterministic, the first \staree query is same irrespective of the graph $G \sim \yesdist \cup \nodist$ that we are querying. By making that query, we get an updated data structure from the oracle and let $\{r,u\}$ be the edge that is labeled with the updated data structure. Then \staralg sets $u$ as the current node.
     \item If the current node $u$ is not a leaf node in $\cT$, \staralg makes a \staree query with a subset $P \subseteq {V(G) \choose 2}$, where $P$ is determined by the label of the node $u$. Note that $P$ satisfies the condition described in Remark~\ref{rem:assumption}. The oracle updates the knowledge 
     structure and \staralg moves to a child of $u$ depending on the updated data structure;
     \item If the current node $u$ is a leaf node in $\cT$, report {\sc YES} or {\sc NO} according to the label of $u$.
 \end{itemize}
 
 Now, we define the notion of \emph{good} and \emph{bad} nodes in $\cT$.  The following definition is inspired from Definition~\ref{defi:bad-dist}.
 \begin{defi}[{\bf Bad node in the decision tree}]\label{defi:badnode}
 Let $u$ be a node of $\cT$ and $(E_Q,V\left(E_Q \right), e,\ell_v)$ be the current data structure. $u$ is said to be \emph{good} if there does not exist $x$ in $V\left(E_Q \right)$ such that $\ell_v(x)=C'$. Otherwise, $u$ is a \emph{bad} node.
 \end{defi}
   
If $G \sim \yesdist$, then \staralg will never encounter a bad node. In other words, when \staralg reaches a bad node of the tree $\cT$, then it can (easily) decide $G \sim \nodist$. However, the inverse in not true. From this fact, consider two claims (Claims~\ref{clm:inter1} and~\ref{clm:inter2}) about the traversal of the decision tree $\cT$ when the graph $G \sim \yesdist \cup \nodist$. These claims  will be useful to show Lemma~\ref{lem:decide-star}. Intuitively, Claim~\ref{clm:inter1} says that the probability of reaching a bad node  is very low when $G \sim \nodist$. Claim~\ref{clm:inter2} says that the probability to reach any particular good node is more when $G \sim \yesdist$ as compared to that of when $G \sim \nodist$.
   
   \begin{cl}[{\bf Probability of reaching a bad node is very low when $G \sim \nodist$}]\label{clm:inter1}
   Let $G \sim \nodist$. Then the probability that \staralg reaches a bad node of the decision tree is $o(1)$. That is, 
   $$\pr_{G \sim \nodist}(\mbox{\staralg  reaches a bad node})=o(1).$$
   \end{cl}
   
   \begin{cl}[{\bf A technical claim to prove Lemma~\ref{lem:decide-star}}]\label{clm:inter2}
    For any good node in the decision tree $\cT$, the following holds:
   $$ \pr_{G \sim \nodist}(\mbox{\staralg reaches } v) \leq \pr_{G \sim \yesdist}(\mbox{\staralg reaches } v).  $$
   \end{cl}
   The proofs of Claims~\ref{clm:inter1} and~\ref{clm:inter2} are non trivial and are in Appendix~\ref{sec:miss-proofs} due to paucity of space.
\remove{In the next section, we first prove , then  Lemma~\ref{lem:decide-star} by using Claims~\ref{clm:inter1} and~\ref{clm:inter2}.}

Now, we will prove Lemma~\ref{lem:decide-star}.
\begin{proof}[Proof of Lemma~\ref{lem:decide-star}] Let $\noleaf$ denote the set of leaf nodes of the decision tree $\cT$ that are labeled $\mbox{{\sc NO}}$. Also, let $\cL_g \subseteq \noleaf$ be the set of leaf nodes that are good and labeled as  $\mbox{{\sc NO}}$.
\begin{align*}
& \pr_{G \sim \nodist}(\mbox{{\staralg (G)} reports {\sc NO}}) \\
&\leq \sum\limits_{v \in \noleaf} \pr_{G \sim \nodist}(\mbox{\staralg$(G)$ reaches $u$})\\
&= \sum\limits_{u \in \cL_g}\pr_{G \sim \nodist}(\mbox{\staralg$(G)$ reaches $u$}) + \sum\limits_{u \in \noleaf \setminus \cL_g}\pr_{G \sim \nodist}(\mbox{\staralg$(G)$ reaches $u$})\\
&\leq \sum\limits_{u \in \cL_g}\pr_{G \sim \nodist}(\mbox{\staralg$(G)$ reaches $u$}) + \pr_{G \sim \nodist}(\mbox{\staralg$(G)$ reaches a bad node})\\
&\leq \sum\limits_{u \in \cL_g}\pr_{G \sim \yesdist}(\mbox{\staralg$(G)$ reaches $u$}) + o(1)\quad\quad(\mbox{By Claims~\ref{clm:inter1} and Claims~\ref{clm:inter2} }) \\
&\leq \pr_{G \sim \yesdist}(\mbox{\staralg$(G)$ reports {\sc NO}}) +o(1)
\end{align*}
%\complain{(Gopi: You sometimes use \staralg and sometimes \staralg$(G)$. What do you want to convey with these diff notations?)}
\end{proof}
\subsection{Proof of Lemma~\ref{theo:clq-count-low}}\label{sec:pf-low}
\noindent
We assume that $t=\omega(\log ^7 n)$. Otherwise, as $\tOm(\cdot)$ hides a multiplicative term of $\frac{1}{\poly (\log n)}$, the stated lower bound is trivial.
Assume, for a contradiction, that there is an algorithm ${\cal A}$ for $t < \frac{m \log n}{8}$ such that
\begin{itemize}
\item it has \ee oracle access to a graph $G_1(V_1,E_1)$ with $\Theta(\sqrt{m})$ vertices and $\Theta(m)$ edges;
\item makes  $o\left(\sqrt{t}\frac{1}{\log^{3.5} n}\right)$ \ee queries;
\item decides whether the number of triangles in $G_1$ is at most $t$ or at least $2t$ with a probability of at least $2/3$.
\end{itemize}
Now we give an algorithm ${\cal A}'$ for 

\begin{itemize}
\item it has \ee oracle access to a graph $G_2(V_2,E_2)$ with $ 4\sqrt{t/\log n}$ vertices and $8t/\log n$ edges, that is, $t=\frac{|E_2| \log n}{8}$;
\item makes  $o\left(\frac{\sqrt{t}}{\log ^{3.5} n}\right)=o\left(\frac{\size{E_2}^{3/2}}{t}\frac{1}{\log ^2 n}\right)$  \ee queries;
\item decides whether the number of triangles in $G_2$ is at most $t$ or at least $2t$ with a probability of at least $2/3$.

\end{itemize}
\paragraph*{Description of ${\cal A}'$ using ${\cal A}$:}
\begin{itemize}
\item Let the unknown graph be $G_1=G_2 \cup G'$ such that $V(G_1)=V(G_2)\sqcup V(G')$ and $E(G_1)=E(G_2)\sqcup E(G')$, where $G'$ is a graph having $\Theta(\sqrt{m-t/\log n})$ vertices (disjoint from $V(G_2)$),  $\Theta\left(\sqrt{(m-8t/\log n)}\right)$  edges, and no triangles~\footnote{The constants in $\Theta(\sqrt{m-8t/\log n})$ and $\Theta(m-8t/\log n)$ are chosen suitably such that the graph $G_1$ satisfies the requirement of $\cA$.}. The number of vertices and edges in $G_1$ are $\Theta(\sqrt{m})$ and $\Theta(m)$, respectively. Also,  the number of triangles in $G_1$ is same as in $G_2$;
\item  As an \ee query to $G_1$ can be answered using one \ee query to $G_2$, we can consider having \ee query access to graph $G_1$;
\item  We run algorithm ${\cal A}$ assuming $G_1$ as the unknown graph;
\item We report the output of ${\cal A}$ as the output of ${\cal A}$.
\end{itemize}
The correctness of ${\cal A}'$ follows from the correctness of ${\cal A}$. The number of queries made by the  
algorithm ${\cal A}'$ is $o\left({\sqrt{t}}\frac{1}{\log^ {3.5} n}\right)$. Recalling that ${\cal A}$ works on graph $G_2(V_2,E_2)$ satisfying $t=\frac{|E_2| \log n}{8}$  and by Lemma~
\ref{theo:clq-count}, algorithm ${\cal A}'$ does not exist as such an algorithm requires at least $\Omega
\left(\frac{\size{E_2}^{3/2}}{t}\frac{1}{\log ^2 n}\right)$ \ee queries, which is $\Omega\left(\frac{{\sqrt{t}}}{\log^ {3.5} n}\right)$. 

Hence, we are done with the proof of Lemma~\ref{theo:clq-count-low}.

\section{Upper bound for estimating triangles using \bis}
\label{sec:ub-triangle}
\noindent
xIn this Section, we prove Theorem~\ref{theo:bis-ub}, which is formally stated as follows:
\begin{theo}[{\bf Upper bound matching the lower bound in Theorem~\ref{theo:tri-lb}}]\label{theo:ub}
There exists an algorithm $\test(G,\eps)$ that has \bis query access to a graph $G(V,E)$ having $n$ vertices, takes a parameter $\eps \in (0,1)$ as input, and reports a $(1 \pm \eps)$-approximation to the number of triangles in $G$ with a probability of at least $1-o(1)$. Moreover, the expected number of \bis queries made by the algorithm is $\tOh\left(\min\{\frac{m}{\sqrt{T}},\frac{m^{3/2}}{T}\}\right)$, where $m$ and $T$ denote the number of edges and triangles in $G$, respectively.
\end{theo}
In order to prove the above theorem, we first prove  Lemmas~\ref{lem:tri-count} and~\ref{lem:tri-count-low}.\remove{ Note that Lemma~\ref{lem:tri-count} and~\ref{lem:tri-count-low} takes a parameter $L$ as input along with approximation parameter $\eps \in (0,1)$, and the approximation guarantee  is also a function of $L$.} 
\begin{lem}[{\bf Upper bound when the number of triangles is \emph{large}}]\label{lem:tri-count}
There exists an algorithm $\alghigh (G,\eps)$ that has \bis query access to a graph $G(V,E)$ having $n$ vertices, $\eps \in (0,1)$ as input, and reports a $(1 \pm \eps)$-approximation to the number of triangles in $G$ with a probability of at least $1-o(1)$. Moreover, the expected number of \bis queries made by the algorithm is $\tOh\left(\frac{m^{3/2}}{T}\right)$, where $m$ and $T$ denote the number of edges and triangles in $G$, respectively.
\end{lem}

\begin{lem}[{\bf Upper bound when the number of triangles is \emph{small}}]\label{lem:tri-count-low}
There exists an algorithm $\alglow (G,\eps)$ that has \bis query access to a graph $G(V,E)$ having $n$ vertices, $\eps \in (0,1)$ as input, and reports a $(1 \pm \eps)$-approximation to the number of triangles in $G$ with a probability of at least $1-o(1)$. Moreover, the expected number of \bis queries made by the algorithm is $\tOh\left(\frac{m+T}{\sqrt{T}}\right)$,  where $m$ and $T$ denote the number of edges and triangles in $G$, respectively. 
\end{lem}

\remove{If $L$ (in the above lemma) is a promised lower bound, then the algorithms $\alghigh (G,\eps)$ and $\alglow (G,\eps)$  produces $(1\pm \eps)$-approximation to the number of triangles in $G$, with a probability of at least $1-o(1)$.}

Our final algorithm $\test(G,\eps)$ (as stated in Theorem~\ref{theo:ub}) is a combination of $\alghigh (G,\eps)$ and $\alglow (G,\eps)$. Informally, $\test(G,\eps)$ calls $\alghigh (G,\eps)$ and $\alglow (G,\eps)$ when $T=\Omega(m)$ and $T=\Oh(m)$, respectively. Observe that, if $\test$ knows $T$ within a constant factor, then it can decide which one to use among  $\alghigh (G,\eps)$ and $\alglow(G,\eps)$. If $\test$ does not know $T$ within a constant factor, then it starts from a guess $L={{n \choose 3}}/{2}$ and updates $L$ by making a \emph{geometric} search until the output of $\test$ is consistent with $L$.  Depending on whether $L =\Omega(m)$ or $L=\Oh(m)$, $\test$ decides which one among $\alghigh (G,\eps)$ and $\alglow(G,\eps)$ to call. This guessing technique is standard in the property testing literature. It has been used several times in the literature (for example in ~\cite{GoldreichR08, EdenLRS15}, generalized in~\cite{EdenRS18}, and used directly in~\cite{DBLP:conf/innovations/AssadiKK19}). So, we explain $\alghigh (G,\eps)$ and $\alglow (G,\eps)$ assuming a promised lower bound on $L$, and the respective query complexities will be in terms of $L$ instead of $T$.

Another important thing to observe is that to execute the above discussed steps of algorithm $\test(G,\eps)$  must know $m$. But we note that an estimate of $m$ will be good enough for our purpose, and that can be estimated by using $\tOh(1)$  \bis queries (see Table~\ref{table:results-bis-is}).

\remove{  In particular, $\test(G,\eps)$ starts by calling $\alghigh (G,\eps)$ with setting $L$ approximately $\frac{m^{3/2}}{2}$. From the output $\hat{t}$ of $\alghigh (G,\eps)$,  $\test(G,\eps)$ decides either $(1 \pm \eps)$-approximation to $T$ or $L$ is not a suitable guess for the lower bound of $T$. In that case, $\test(G,\eps)$ reduces $L$ by a constant factor, and call $\alghigh (G,\eps)$ with updated $L$. Note that after some iteration, ether $\test(G,\eps)$ must have quit by reporting an $(1 \pm \eps)$-approximation to $T$ or $L=\Oh(m)$. In the later case, call $\alglow (G,\eps)$, and proceed as we were doing with $\alghigh (G,\eps)$. Note that the total number of iterations will be $\Oh(\log n)$, and the query complexity will be dominated by the query complexity of the last iteration (when $L=\Theta(T)$).} 

In Appendix~\ref{sec:bis-pre}, we discuss some properties of \bis and some tasks it can perform. These properties will be useful while describing our $\alghigh (G,\eps)$ and $\alglow (G,\eps)$, and proving Lemma~\ref{lem:tri-count} and Lemma~\ref{lem:tri-count-low}, in Section~\ref{sec:alg-high} and Section~\ref{sec:alg-low}, respectively. %Finally, we describe algorithm  $\test(G,\eps)$ and prove Theorem~\ref{theo:ub} in Section~\ref{sec:ub-final}.

\subsection{Some preliminaries about \bis}\label{sec:bis-pre}
\noindent
Let $G(V,E)$ be the unknown graph to which we have \bis query access. One can compute the exact number of edges using $\tOh(\size{E(G)})$ queries~\cite{DBLP:journals/talg/BeameHRRS20} deterministically. Also, we can estimate the number of edges in  graph $G$~\cite{DBLP:journals/talg/BeameHRRS20,abs-1908-04196,DellLM20-soda} and sample an edge  from $G$ \emph{almost uniformly}~\cite{DellLM20-soda}, with a  probability of at least $1-o(1)$, and making $\tOh(1)$ \bis queries. Here, we would like to note that, all three results we mentioned above hold for induced subgraphs as well as induced bipartite subgraph, as formally described below. Those will be used when we design our upper bounds in Appendix~\ref{sec:alg-high} and~\ref{sec:alg-low}.

\begin{pro}[{\bf Exact edge estimation using \bis}~\cite{DBLP:journals/talg/BeameHRRS20}]\label{pro:bis-ub1}
There exists a deterministic algorithm that has \bis query access to an unknown graph 
$G(V,E)$ with $n$ vertices, takes $X \subseteq V(G)$ (alternatively, two disjoint subsets $A,B$) as input, makes $\tOh(\size{E(G[X])})$ (alternatively,  $\tOh(\size{E(G[A,B])})$) \bis queries, and reports all the edges in ${E(G)}$ (alternatively, ${E(G[A,B])}$).
\end{pro}
\begin{pro}[{\bf Approximate edge estimation using \bis}~\cite{abs-1908-04196,DellLM20-soda}]\label{pro:bis-ub2}
There exists an algorithm that has \bis query access to an unknown graph 
$G(V,E)$ with $n$ vertices, takes $X \subseteq V(G)$ (alternatively, two disjoint subsets $A,B$) and a parameter $\eps \in (0,1)$ as inputs, makes $\tOh(1)$ \bis queries, and reports a $(1\pm \eps)$-approximation to $\size{E(G[X])}$ (alternatively, $\size{E(G[A,B])}$), with a probability of at least  $1-o(1)$.
\end{pro}
To state the next proposition, we need the following definition.

\begin{defi}[{\bf Approximate uniform sample from a set}]\label{defi:apx-samp}
 For a nonempty set $X$ and $\eps \in (0,1)$, getting a $(1 \pm \eps)$-approximate uniform sample from $X$ means getting a sample from a distribution on $X$ such that the probability of getting $x \in X$  lies in $[(1-\eps)/\size{X},(1+\eps)/{\size{X}}]$.
 \end{defi}

\begin{pro}[{\bf Approximate edge sampling using \bis}~\cite{DellLM20-soda}]\label{pro:bis-ub3}
There exists an algorithm that has \bis query access to an unknown graph 
$G(V,E)$ with $n$ vertices, takes $X \subseteq V(G)$ (alternatively, two disjoint subsets $A,B$) and a parameter $\eps \in (0,1)$ as inputs, makes $\tOh(1)$ \bis queries, and reports a $(1\pm \eps)$-approximate uniform sample from $E(G[X])$ (alternatively, $E(G[A,B]$), with a probability of at least  $1-o(1)$.
\end{pro}
Observe that the following corollary follows from  Propositions~\ref{pro:bis-ub1},~\ref{pro:bis-ub2} and~\ref{pro:bis-ub3}, by taking $A=\{v\}$ and $B=Z$, where $v \in V(G)$ and $Z \subseteq V(G) \setminus \{v\}$.
\begin{coro}[{\bf \bis query can extract useful information about the neighborhood of a given vertex}]\label{coro:bis}
$~~~~~~~~~~~~~~~~$
\begin{description}
\item[(i) Entire neighbourhood of a vertex using \bis :] There exists a deterministic algorithm that has \bis query access to an unknown graph 
$G(V,E)$ with $n$ vertices, takes $v \in V(G)$ and $Z \subseteq V(G)\setminus \{v\}$ as input, makes $\tOh(\size{N_{G}(v) \cap Z})$ \bis queries, and reports all the neighbors of $v$ in $Z$.
\item[(ii) Approximate degree using \bis :] There exists an algorithm that has \bis query access to an unknown graph 
$G(V,E)$ with $n$ vertices, takes $v \in V(G)$, $Z \subseteq V(G)\setminus \{v\}$ and $\eps \in (0,1)$ as inputs, makes $\tOh(1)$ \bis queries, and reports a $(1 \pm \eps)$-approximation to $\size{N_{G}(v) \cap Z}$, with a probability of at least $1-o(1)$.
\item[(iii) Finding an approximate random neighbor of a vertex using \bis :]  There exists an alg- orithm that has \bis query access to an unknown graph 
$G(V,E)$ with $n$ vertices, takes $v \in V(G)$, $Z \subseteq V(G)\setminus \{v\}$ and $\eps \in (0,1)$ as inputs, makes $\tOh(1)$ \bis queries, and reports a $(1 \pm \eps)$-approximate uniform sample from the set ${N_{G}(v) \cap Z}$, with a probability of at least  $1-o(1)$.
\end{description}
\end{coro}
\subsection{Algorithm \alghigh and proof of Lemma~\ref{lem:tri-count}}\label{sec:alg-high}
\noindent
Algorithm \alghigh is inspired by the triangle estimation algorithm of Assadi \etal ~\cite{DBLP:conf/innovations/AssadiKK19}~\footnote{Actually, they have given an algorithm for estimating the number of copies of any given subgraph of fixed size.} when we have the following query access to the unknown graph.

\begin{description}
\item[{\sc Adjacency Query}:] Given vertices $u,v \in V(G)$ as input, the oracle reports whether $(u,v)$ is an edge or not;
\item[{\sc Degree Query}:] Given  a vertex $u \in V(G)$ as input,  the oracle reports the degree of vertex $u$ in $G$;
\item[{\sc Random Neighbor Query}:] Given  a vertex $u \in V(G)$,  the oracle reports a neighbor of $u$ uniformly at random if the degree of $u$ is nonzero. Otherwise, the oracle reports a special symbol $\perp$;
\item[{\sc Random Edge Query:}] With this query, the oracle reports an edge from the graph $G$ uniformly at random. 
\end{description}

The number of queries to the oracle made by Assadi \etal's algorithm~\cite{DBLP:conf/innovations/AssadiKK19} is $\tOh\left(\frac{m^{3/2}}{L}\right)$, where $m$ denotes the number of edges and $L$ is a promised lower bound on the number of triangles in $G$. Also, note that, the triangle estimation algorithm by Assadi \etal ~\cite{DBLP:conf/innovations/AssadiKK19} can be \emph{suitably} modified even if we have approximate versions of {\sc Degree}, {\sc Random Neighbor} and {\sc Random Edge} queries, as described below.

%Now consider the approximate variants of {\sc Apx Degree Query}, {\sc Apx Random Neighbor Query} and {\sc Random Neighbor Query}.
\begin{description}
\item[{\sc Apx Degree Query}:] Given  a vertex $u \in V(G)$ and $\eps \in (0,1)$  as input,  the oracle reports a $(1 \pm \eps)$-approximation to the degree of  vertex $u$ in $G$;
\item[{\sc Apx Random Neighbor Query}:] Given  a vertex $u \in V(G)$ and $\eps \in (0,1)$  as input, the oracle reports a $(1 \pm \eps)$-approximate uniform sample from $N_G(u)$ if the degree of $u$ is nonzero. Otherwise, the oracle reports a special symbol $\perp$;
\item[{\sc Apx Random Edge Query:}] Given $\eps \in (0,1)$, the oracle reports a $(1 \pm \eps)$-approximate uniform sample from $E(G)$.
\end{description}

From Corollary~\ref{coro:bis}, $\tOh(1)$ \bis queries are enough to simulate 
{\sc Apx Degree Query} and {\sc Apx Random Neighbor Query}, with a probability of at least $1-o(1)$. Also, by Proposition~\ref{pro:bis-ub3}, {\sc Apx Random Edge Query} can be simulated by $\tOh(1)$ \bis queries, with a probability of at least $1-o(1)$. Moreover, a \bis query can trivially simulate an {\sc Adjacency Query}. Combining these facts with the fact that the triangle estimation algorithm by Assadi \etal ~\cite{DBLP:conf/innovations/AssadiKK19} can be {suitably} modified even if we have approximate versions of {\sc Degree}, {\sc Random Neighbor} and {\sc Random Edge} queries, we are done with the proof of Lemma~\ref{lem:tri-count}. 
\subsection{Algorithm \alglow and the proof of Lemma~\ref{lem:tri-count-low}}\label{sec:alg-low}
\noindent
\remove{We assume that $L$ in Lemma~\ref{lem:tri-count-low} is at least $\poly (\log n)$. Otherwise, we find we find all the edges in graph $G$, and hence can compute the exact number of triangles in $G$. Note that, this is possible by Proposition~\ref{pro:bis-ub1} by making $\tOh(m)=\tOh\left(\frac{m}{\sqrt{L}}\right)$  \bis queries.
 Algorithm \alglow starts by finding $\widehat{m}$, an $\left(1 \pm \Oh({\eps})\right)$-approximation to $m=\size{E(G)}$, with a probability of at least $1-o(1)$. By Proposition~\ref{pro:bis-ub2}, this is possible by making $\tOh(1)$  \bis queries. If $\widehat{m}=\tOh(1)$, then we find all the edges in graph $G$ by making $\tOh(m)=\tOh(1)$ \bis queries, and hence compute the exact number of triangles in $G$. Note that, this is possible by Proposition~\ref{pro:bis-ub1}. So, we assume that there are at least $\poly (\log n)$ edges in graph $G$. 
 }
  Algorithm \alglow is inspired by the streaming algorithm for triangle counting by McGregor \etal~\cite{DBLP:conf/pods/McGregorVV16}. Algorithm \alglow extracts a subset of edges by making \bis queries in a specific way as explained below. Later, we discuss that those sets of edges will be enough to estimate the number of triangles in $G$. 
  
\paragraph*{Generating a random sample $S \subseteq V(G)$ and exploring its neighborhood:\\}
\quad Algorithm \alglow  adds each vertex in $V(G)$ to $S$ with probability $\tOh\left(\frac{1}{\sqrt{L}}\right)$. Recall  Corollary~\ref{coro:bis} (i). For each $v \in S$, we find all the neighbors of $v$ (the set $N_G(v)$) by making $\tOh(\size{N_{G}(v)})$ \bis queries. Let $E_S$ be the set of all the edges having at least one end point in $S$, that is, $E_S=\{(v,w):v \in S~\mbox{and}~w \in N_{G}(v)\}$. After finding $E_S$, we do the following. For each $v \in S$, we find all the edges in the subgraph induced by $N_{G}(v)$ by using $\tOh\left(\size{E(G[N_G(v)])}\right)$ \bis queries. This is again possible by Corollary~\ref{coro:bis} (i). Note that $\size{E(G[N_G(v)])}$ is the number of edges in the subgraph of $G$ induced by $N_G(v)$. Let $E_S'$ be the set of all edges present in the subgraph  induced by $N_{G}(v)$ for some $v \in S$, that is, $E_S'=\bigcup\limits_{v \in S}E(G[N_G(v)])$. Later, we argue that the number of  \bis queries that we make to generate $E_S$ and $E_S'$ is bounded in expectation.

Apart from $S, E_S$ and $E_S'$, \alglow extracts some more required edges by making \bis queries, as explained below.
% to estimate the number of triangles.

\paragraph*{Generating $F$, a set of $\left(1 \pm \Oh{(\eps)}\right)$-approximate uniform sample from $E(G)$, and exploring the subgraphs induced by sets $N_{G}(v) \cap V(F)$ for each $v \in V(F)$:\\}
 Algorithm \alglow calls the algorithm corresponding to Proposition~\ref{pro:bis-ub3}, for $\tOh\left(\frac{ {m}}{\sqrt{L}}\right)$ times. By this process, we get a set $F$  of $\left(1 \pm \Oh({\eps})\right)$-approximate uniform sample from $E(G)$, with a probability of at least $1-o(1)$. Note that $\size{F}=\tOh\left(\frac{ m}{\sqrt{L}}\right)$, and the number of \bis queries we make to generate $F$ is $\tOh\left(\frac{m}{\sqrt{L}}\right)$. Let $V(F)$ be the set of vertices present in at least one edge in $F$. For each vertex $v \in V(F)$, we find all the edges in the subgraph of $G$ induced by $N_G(v) \cap V(F)$, by using $\tOh\left(\size{E(G[N_G(v) \cap V(F)])}\right)$ \bis queries (see Corollary~\ref{coro:bis} (i)). Note that $\size{E\left(G[N_G(v) \cap V(F)]\right)}$ is the number of edges in the subgraph of $G$ induced by $N_G(v) \cap F$\remove{, or equivalently $T_v$, the number of triangles having $v$ as one of the vertex}. Let $E_{F}$ be the set of all the edges that are present in  subgraphs induced by $N_G(v) \cap V(F)$ for some $v \in V(F)$, that is, 
 $E_F=\bigcup\limits_{v \in V(F)} E(G[N_{G}(v) \cap V(F)])$. Later we show that the expected number of \bis queries needed to find $F$ and $E_F$ is bounded.
 
 In algorithm \alglow, \bis queries are made only to generate $S,E_S,E_S',F$ and $E_F$. After these sets are generated, no more \bis queries are made by the algorithm. We formally prove the query complexity of \alglow. But, first, we show that $S, E_S, E_S', F$ and $E_F$ can be carefully used to estimate $T$, the number of triangles in $G$. 
 
 \paragraph*{Connection with streaming algorithm for triangle counting by McGregor \etal~\cite{DBLP:conf/pods/McGregorVV16}:\\}
 (Estimating the number of triangles from $S, \, E_S, \, E_S', \, F$ and $E_F$)
 
 McGregor \etal~\cite{DBLP:conf/pods/McGregorVV16} gave a two-pass algorithm that estimates the number of triangles in a graph $G$ when the edges of $G$ arrive in an arbitrary order. Moreover, the space complexity of their algorithm is $\tOh\left(\frac{m}{\sqrt{L}}\right)$. Note that their algorithm assumes a lower bound  $L$ on the number of triangles in the graph. The high level sketch of their algorithm is as follows:
 \begin{itemize}
 \item Generate a subset $X$ of $V(G)$ by sampling each vertex in $V(G)$ with probability $\tOh\left(\frac{1}{\sqrt{L}}\right)$;
 \item In the first pass, the edges having at least one vertex in $X$ is found, and let it be $E_X$. Also, in the first pass, a subset of edges $Y$ is generated by sampling each edge with probability $\tOh\left(\frac{1}{{\sqrt{L}}}\right)$;
 \item In the second pass, for each edge $e=\{x,y\}$ in the stream, their algorithm finds the vertices in $X$ with which $e$ forms a triangle. Also, for each edge, $e=\{x,y\}$ in the stream, their algorithm finds the pairs of edges in $Y$ that forms a triangle with $e$. Let $Z$ be the set of \emph{useful} edges in the second pass, that is, the set of edges that either forms a triangle with a vertex in $X$ or forms a triangle with two edges in $Y$. 
 
\end{itemize}  
Note that their algorithm does not talk about the set $Z$. We are introducing it for our analysis. Executing the two passes described above is straightforward. They have proved that performing two passes as described is good enough to estimate the number of triangles in the graph. 
 
 Now, we compare the information maintained by our algorithm with the information
 maintained by McGregor \etal's algorithm. $S$ and $E_S$ in our algorithm \alglow follows the same probability distribution as that of $X$ and $E_X$, respectively, in McGregor \etal 's algorithm. Recall that $F$ is a set of $\tOh\left(\frac{m}{\sqrt{L}}\right)$ $(1\pm \eps)$-approximate sample from $E(G)$ with a probability $1-o(1)$. But $Y$ in McGregor \etal 's algorithm is generated by sampling each edge with probability $\tOh\left(\frac{1}{{\sqrt{L}}}\right)$. But observe that the \emph{total variation} distance between the probability distributions of $F$ and $Y$ is $o(1)$. 

Before discussing about $E_S'$ and $E_F$ in our algorithm \alglow, consider the following observation about the set $Z$. Note that we have defined $Z$ while describing the second pass of McGregor \etal 's algorithm.

\begin{obs}
Consider $X$, $E_X$, and $Y$ generated by the first pass of McGregor \etal 's algorithm. Let $E_X'$ be the edges in the subgraph induced by $N_G(v)$ for some $v \in X$, and let $E_Y$ be the set of edges in the subgraph induced by $N_G(v) \cap V(Y)$. Here $V(Y)$ denotes the set of vertices present in at least one vertex of $Y$. Then for each edge $e \notin E_X' \cup E_Y$, there is neither a vertex in $X$ with {which} $e$ forms a triangle in $G$ nor there are two edges in $Y$ with which $e$ forms a triangle in $G$. Then the set of useful edges $Z$ is $E_X' \cup E_Y$.
\end{obs}

By the above observation, $E_S'$ and $E_F$ in our algorithm \alglow are essentially enough for maintaining the information and executing the same steps as that of the second pass of McGregor \etal 's algorithm.

Putting everything together, algorithm \alglow outputs a $\left(1 \pm \eps\right)$-approximation to the number of triangles in the graph.

 \paragraph*{Query complexity analysis:\\} 
The set $S$ can be generated without making any \bis queries.
The number of \bis queries we make to find the set $E_S$ is at most $\sum\limits_{v \in S}\tOh(\size{N_G(v)})$, which in expectation is 
$$ \E\left[\sum\limits_{v \in S}\tOh(\size{N_G(v)})\right]=\sum\limits_{v \in V(G)}\Pr(v \in S)\cdot \tOh(\size{N_G(v)})=\tOh \left(\frac{m}{\sqrt{L}}\right).$$ 

The number of \bis queries we make to find the set $E_S'$ is at most $\sum\limits_{v \in S}\tOh\left(\size{E(G[N_{G}(v)])}\right)$. Note that $\size{E(G[N_{G}(v)]}$ is $T_v$, that is, the number of triangles having $v$ as one of the vertex. So, the expected number of \bis queries we make to find $E_S'$ is at most

$$ \E\left[\sum\limits_{v \in S}\tOh(T_v)\right]=\sum\limits_{v \in V(G)} \Pr(v \in S)\cdot \tOh(T_v) =\tOh\left(\frac{T}{\sqrt{L}}\right).$$
The number of \bis queries we make to find the set $F$ is $\tOh(\size{F})=\tOh\left(\frac{m}{\sqrt{L}}\right)$

The number of \bis queries to generate $E_F$ is at most $\sum 
\limits_{v \in V(F)} \tOh\left(\size{E(G[N_{G}(v) \cap V(F)])}\right)$. 
Observe that an edge $\{x,y\}$ is present in $E_F$ if there exists a $z 
\in V(G)$ such that $\{x,y,z\}$ forms a triangle in $G$ and $\{x,z\}$ 
and $\{y,z\}$ are in $F$. So, the probability that an edge $\{x,y\}$ is 
in $E(G[N_{G}(v) \cap V(F)]$ is at most $\size{\Gamma(\{x,y\})}\cdot 
\tOh\left(\frac{1}{{L}}\right)$, where $\Gamma(\{x,y\})$ denotes the set of common neighbors of  $x$ and $y$ in $G$. So, the expected number of \bis 
queries to enumerate all the edges in $E_F$ is at most 
 $$\sum\limits_{\{x,y\}\in E(G)}\tOh \left(\frac{\size{\Gamma(\{x,y\}})}{\sqrt{L}}\right) = \tOh \left( \frac{T}{\sqrt{L}}\right).$$
 
 Hence, the expected number of \bis queries made by the algorithm is $\tOh\left(\frac{m+T}{\sqrt{L}}\right)$.

% \paragraph*{Query complexity analysis of the algorithm:}

%\input{ub-final}

\section{Conclusion}
\noindent
We touched upon two open questions of Beame \etal~\cite{DBLP:journals/talg/BeameHRRS20} in this paper. We  resolved the query complexity of triangle estimation when we have a \bislong oracle access to the unknown graph when $T =\Omega(m)$. But the query complexity of triangle counting remain illusive when $T=o(m)$ though we believe that our upper bound of $\tOh(m/\sqrt{T})$ \bis queries is tight in this regard. It is also interesting if our upper bound can be improved when $T=o(m)$. 

\newpage
\bibliographystyle{alpha}
\bibliography{reference}
\newpage
\appendix

\section{Missing proofs from Section~\ref{sec:lb-clique}}\label{sec:miss-proofs}

%\subsection{Proof of Observation~\ref{obs:edge}}\label{app:obs-edge}

\begin{obs}[Observation~\ref{obs:edge} restated]
\begin{enumerate}
\item[(i)] For $G\sim \yesdist \cup \nodist$, then the number of vertices in $G$ is $4\sqrt{m}$. Also, $\frac{\sqrt{m}}{2} \leq \size{A}, \size{B}, \size{C},\size{D} \leq 2\sqrt{m}$ with a probability of at least $1-o(1)$, and the number of edges in $G$ is $\Theta(m)$ with a probability of at least $1-o(1)$;
\item[(ii)] If $G\sim \yesdist$, then there are at most  $t$ triangles in $G$ with a probability of at least $1-o(1)$,
\item[(iii)] If $G\sim \nodist$, $ \frac{8t}{m}\leq \size{C'}\leq \frac{32t}{m}$ with a probability of at least $1-o(1)$, and there are at least $2t$ triangles in $G$ with a probability of at least $1-o(1)$.
\end{enumerate}
\end{obs}
\begin{proof}
\begin{itemize}
\item[(i)] By the construction of $G\sim \yesdist \cup \nodist$, the number of vertices in any such $G$ is $4\sqrt{m}$. 

For the cardinalities of $A,B,C,D$, we only argue for $|A|$. For $|B|,|C|,|D|$, the proofs are similar. Note that each vertex in $V(G)$ is put into one of $A,B,C,D$ uniformly at random, and independent of other vertices. So, $\E[\size{A}]=\sqrt{m}$. By using Chernoff-Hoeffding bound (See Lemma~\ref{lem:cher_bound1} in Appendix~\ref{sec:prelim}), we have 
$$\pr\left(\size{\size{A}-\sqrt{m}}\geq \frac{\sqrt{m}}{2} \right) \leq \exp{\left(- \frac{(\sqrt{m}/2)(1/2)^2}{3}\right)} \leq o(1).$$
The last inequality holds as $m=\Omega(\log ^2 n).$ Note that the above equation implies that 
$$\pr\left(\frac{\sqrt{m}}{2} \leq \size{A} \leq 2\sqrt{m}\right)\geq 1-o(1).$$

\item[(ii)] Here we work on the conditional space that $\frac{\sqrt{m}}{2} \leq \size{A}, \size{B}, \size{C},\size{D} \leq 2\sqrt{m}$.

 If we set a indicator random variable for each triple of vertices (one in each of $A,B,C$) such that it is set to $1$ if the tree vertices forms a triangle in $G$.

Let $N$ be the number of indicator random variables. Observe that 
 $\frac{m^{3/2}}{8}\leq N \leq 8m^{3/2}$. Due to our construction of $G \sim \yesdist$, the probability that each such indicator variable takes vale $1$ is $p=\left(\sqrt{\frac{t}{16m^{3/2}}}\right)^2=\frac{t}{16m^{3/2}}$, and each such indicator random variable may depend on at most $d \leq 6\sqrt{m}$ other random variables. Taking $X$ as the sum of $N$ indicator random variables, 
 $$\E[X]=Np \leq \frac{t}{2}.$$
 
  Setting $\delta=\frac{t}{2}$ and applying Lemma~\ref{lem:depend:high_prob}, we get

\begin{eqnarray*}
\pr(X \geq t) &=& \pr(X \geq \E[X]+\delta)\\  
&\leq& \exp{\left(-\frac{\delta^2\left(1-\frac{d+1}{4N}\right) }{ 2(d+1)\left(Np+\frac{\delta}{3}\right)}\right)}\\
&\leq & \exp{\left(-\frac{(t/2)^2\left(1-\frac{6\sqrt{m}+1}{m^{3/2}/2}\right) }{ 2(6 \sqrt{m}+1)\left(t/2+{t}/{6}\right)}\right)}\\
&=& o(1).
\end{eqnarray*}
\item[(iii)] Here we work on the conditional space that $\frac{\sqrt{m}}{2}\leq \size{C}\leq 2 \sqrt{m}$. 

As we put each vertex in $C$ into $C'$ with probability $\frac{32t}{m^{3/2}}$, 
$$\frac{16t}{m}\leq \E[C'] \leq \frac{64t}{m}.$$

By using Chernoff-Hoeffding bound (See Lemma~\ref{lem:cher_bound1} in Appendix~\ref{sec:prelim}), we have 

$$\pr\left(\frac{8t}{m}\leq \size{C'}\leq \frac{32t}{m}\right)\geq 1-o(1).$$

By the construction of $G \sim \nodist$, the number of triangles in $G$ is at least $\size{C'}|A||B|$, which is at least $2t$ with a probability of at least $1-o(1)$.

\end{itemize}
\end{proof}

%\begin{obs}[Observation~\ref{obs:lem_eqv_aug} restated]
%Let $G \sim  \yesdist \cup \nodist$. Each \ee query to $G$ can be %simulated by using an \staree query 

\subsection*{Proofs of Claims~\ref{clm:inter1}~and~\ref{clm:inter2}}
To prove Claim~\ref{clm:inter1}, we need Observations~\ref{obs:cond} and~\ref{obs:cond1}. Informally speaking, these two observations help us to argue that enough randomness is left in the unknown part of the graph when we make $q=o\left(\frac{m^{3/2}}{t}\frac{1}{\log ^2 n}\right)$ \staree queries. Particularly, Observation~\ref{obs:cond} says that, when \staralg is at a good node $u$ of the decision tree, then  a graph $G \sim \yesdist$ can be generated respecting the current data structure that we have at node $u$.  Observation~\ref{obs:cond1} says that, when \staralg is at some node $u$ of the decision tree, then  a graph $G \sim \nodist$ can be generated respecting the current data structure that we have at node $u$. The proof of Claim~\ref{clm:inter2} uses the connection between graphs in $\yesdist$ and  $\nodist$, as described in Remark~\ref{rem:yes-no}.
\begin{obs}[{\bf A graph $G \sim \yesdist$ can be generated respecting any good node in the decision tree}]\label{obs:cond}
Let $u$ be the current node of the decision tree $\cT$ that is good and $(E_Q,V\left(E_Q \right), e,\ell_v)$ be the current data structure. Then a graph $G \sim \yesdist$ can be generated as follows conditioned on the fact that \staralg{}$(G)$ reaches $u$. 
\begin{itemize}
    \item For each $x\in V\left(E_Q \right)$, put $x$ in {the vertex partition indicated by $\ell_v(x)$}. Put each vertex in $V(G) \setminus V\left(E_Q \right)$ to one of the  parts out of $A, \, B, \, C, \, D$ uniformly at random and independent of other vertices;
    \item  For each $\{x,y\} \in E_Q $, add an edge between $x$ and $y$ if and only if $e(\{x,y\})=1$. Then for each $\{x,y\} \in {V(G) \choose 2} \setminus E_Q$, do the following:
    
\begin{itemize}
\item Add an edge between $x$ and $y$ if one vertex is in $A$ and the other is in $B$;
\item Add an edge between $x$ and $y$ if one vertex is in $C$ and the other is in $D$;
\item If one vertex out of $x$ and $y$ is in $A \cup B$ and the other in $D$, then $\{x,y\}$ does not form an edge.
\item  Add an edge between $x$ and $y$ with probability $\sqrt{\frac{t}{16 m^{3/2}}}$ if one vertex in $A \cup B$ and the other in $C$.

\end{itemize}    
\end{itemize}

\end{obs}
\begin{obs}[{\bf {\bf A graph $G \sim \nodist$ can be generated respecting any node in the decision tree}}]\label{obs:cond1}
Let $u$ be the current node of the decision tree $\cT$ and $(E_Q,V\left(E_Q \right), e,\ell_v)$ be the current data structure. Then a graph $G \sim \nodist$ can be generated as follows conditioned on the fact that \staralg $(G)$ reaches $u$. 
\begin{itemize}
    \item For each $x\in V\left(E_Q \right)$, put $x$ in {the vertex partition indicated by $\ell_v(x)$}.  Put each vertex in $V(G) \setminus V\left(E_Q \right)$ to one of the  parts out of $A,B,C,D$ uniformly at random and independent of other vertices;
    \item  For each $\{x,y\} \in E_Q $, add an edge between $x$ and $y$ if and only if $e(\{x,y\})=1$. Then for each $\{x,y\} \in {V(G) \choose 2} \setminus E_Q$, do the following:
    
\begin{itemize}
\item Add an edge between $x$ and $y$ if one vertex is in $A$ and the other is in $B$;
\item Add an edge between $x$ and $y$ if one vertex is in $C$ and the other is in $D$;
\item If one vertex out of $x$ and $y$ is in $A \cup B$ and the other is in $D$, then $\{x,y\}$ does not form an edge.
\item  Add an edge between $x$ and $y$ with a probability of at least $\sqrt{\frac{t}{16 m^{3/2}}}$ if one vertex is in $A \cup B$ and the other is in $C$;
\item For each vertex $x \in C \setminus C'$, add $x$ to $C'$ with probability $32/m^{3/2}$.  Add each edge {of the form} $\{{x,y}:x \in A\cup B, y \in C'\}$ to the graph;
\end{itemize}    
\end{itemize}

\end{obs}
Observations~\ref{obs:cond} and~\ref{obs:cond1} follow from the descriptions of our hard distributions ($\yesdist$ and $\nodist$ in Section~\ref{sec:hard}), along with the description of \staree oracle and its interplay with algorithm \staralg.
 Now we will prove Claims~\ref{clm:inter1} and~\ref{clm:inter2} by using Observations~\ref{obs:cond} and~\ref{obs:cond1}.
   \begin{cl}[Claim~\ref{clm:inter1} restated]
   Let $G \sim \nodist$. Then the probability that \staralg reaches a bad node of the decision tree is $o(1)$. That is, 
   $$\pr_{G \sim \nodist}(\mbox{\staralg  reaches a bad node})=o(1).$$
   \end{cl}
 \begin{proof}
 Recall that $q=o\left(\frac{m^{{3/2}}}{t}\frac{1}{\log ^2 n}\right)$ and \staralg makes (at most) $q$ \staree queries. The execution of \staralg starts from the root node of the decision tree $\cT$, which is trivially a good node. For $i \in [q]$, let $u_i$ be the node of $\cT$ that \staralg reaches after making $i$ \staree queries. Assume that $u_0$ is the root of the tree $\cT$. Note that $u_i$ is a child of $u_{i-1}$ for each $i \in [q]$.
 Claim~\ref{clm:inter1} says that, there exists an $i \in [q]$ such that $u_i$ is a bad node, with a probability of at most $o(1)$. Observe that we will be done with the proof by showing the following: 
 if $u_0,\ldots,u_k$ are good nodes with $k \leq q-1$, then $u_{k+1}$ is a bad node with a probability of at most $\Oh\left(\frac{t \log ^2 n}{m^{3/2}}\right)$.
 
 As $k \leq q-1$, $u_k$ is not a leaf node. Let {the label $u_k$ be $P$}. Note that $P$ is a subset of ${V(G) \choose 2}$ with which \staralg makes \staree query. Let $(E_Q,V\left(E_Q \right), e,\ell_v)$ be the current data structure of \staralg after making \staree query with input $P$. %Consider $V(P)$, \complain{the set of vertices present in at least one pair of vertices in $P$ (Gopi: is this not well defined if you just write $V(P)$?)}.
   Recall the definition of a bad node (Definition~\ref{defi:badnode}). Observe that $u_{k+1}$ is a bad node if and only if $V(P) \cap C' \neq \emptyset$.
 \remove{
 Let $I_{AB}$ be the set of edges between $A$ and $B$; $I_{A'B'}$ be the set of edges between $A'$ and $B'$; $\Bar{I}_{AA'}$ be the set of non edges between $A$ and $A'$; $\Bar{I}_{BB'}$ be the set of non edges between $B$ and $B'$. Note that if $G \sim \yesdist$, then the sets $I_{AB}$, $I_{A'B'}$, $\Bar{I}_{AA'}$ and $\Bar{I}_{BB'}$ are empty sets. Also, if $G \sim \nodist$, then $I_{AB}$, $I_{A'B'}$, $\Bar{I}_{AA'}$ and $\Bar{I}_{BB'}$ are nonempty and have the same cardinality; Moreover, each of the four sets has at least $t/\sqrt{m}$ elements and at most $4t/\sqrt{m}$ elements with a probability of at least $1-o(1)$.  
 
  The analysis is divided into two following cases depending on how \staree oracle responds to the \staree query with $P$ input by \staralg.
 
First consider the case when $\size{P}\leq \tau$.
 if there exists two distinct vertices $x$ and $y$ in 
 $K$ such that $\ell_e(\{x,y\})=\mbox{bad}$. Let $Q=\{\{x,y\}:x \in X~\mbox{and}~y \in Y\}$. From the definition of $\ell_e$, we can say $u_{k+1}$ is a bad node if 
 $Q\cap (I_{AB}\cup I_{A'B'}\cup I_{AA'}\cup I_{BB'})$ is nonempty.
 } 
 Hence, when $|P|\leq \tau$, we can deduce the following by Observations~\ref{obs:cond} and~\ref{obs:cond1}:
 
\begin{align*}
& \pr_{G \sim \nodist}(u_{k+1}~\mbox{is bad}~|~u_0,\ldots,u_k~\mbox{are good}) \\
&\leq \pr_{G \sim \nodist}(V(P) \cap C' \neq \emptyset~|~u_0,\ldots,u_k~\mbox{are good}) \\
&\leq  \pr(V(P)\cap C' \neq \emptyset~|~u_0,\ldots,u_k~\mbox{are good})  \\
%&\leq& \frac{V(P)}{4\sqrt{m}}\size{C'} \\
&\leq \size{V(P)} \times \frac{1}{4} \frac{32t}{m^{3/2}} \\  
&\leq {2\tau} \cdot \frac{8t}{m^{3/2}}=\Oh \left(\frac{t\log ^2 n}{m^{3/2}}\right)\quad\quad\quad\quad(\because |P|\leq \tau ~\mbox{and}~\tau = 25 \log^2 n.)
\end{align*}

Now consider the case when $\size{P}>\tau$. In this case, recall the behavior of \staree oracle. Let $P'\subseteq P$ be generated uniformly at random such that $\size{P'}=\tau$.  In this case, if there exists vertex pair in $P'$ that forms an edge in $G$, then $u_{k+1}$ is a bad node if and only if $V(P') \cap C' \neq \emptyset$. Hence, we can deduce the following by Observation~\ref{obs:cond} and~\ref{obs:cond1}:

\begin{align*}
&\pr_{G \sim \nodist}(u_{k+1}~\mbox{is bad}~|~u_0,\ldots,u_k~\mbox{are good}) \\
&\leq \pr_{G \sim \nodist}(V(P) \cap C' \neq \emptyset~|~u_0,\ldots,u_k~\mbox{are good}) + \pr \left(\{x,y\} \notin E(G)~ \forall\{x,y\} \in P'\right)
%&&~~~~~~~~~~~~~~~+ \pr_{G \sim \nodist}(Q'=\emptyset~|~u_0,\ldots,u_k~\mbox{are good}) \\
\end{align*}

  The first term can be bounded by $o \left( \frac{t \log ^2 n}{m^{3/2}}\right)$ in the similar fashion as we did when $\size{P}\leq \tau$. To finish the proof we need to show that 
  \begin{equation}\label{eqn:smallo}
    \pr_{G \sim \nodist}(\{x,y\} \notin E(G)~ \forall\{x,y\} \in V(P')~|~u_0,\ldots,u_k~\mbox{are good}) =o(1).
  \end{equation}

  Consider a particular $\{x,y\} \in P'$. By Remark~\ref{rem:assumption} (i), the labels of at most one vertex out of $x$ and $y$ is known till now. By Observation~\ref{obs:prob-edge}, the probability that $\{x,y\}$ is an edge is at least $1/4$. Note that $\size{P'}=\tau=25 \log ^2 n$. So, there are either $5 \log n$ pairwise disjoint vertex pairs in $P'$ or $5 \log n$ vertex pairs having a common vertex in $P'$~\cite{DBLP:series/txtcs/Jukna11}~\footnote{\cite{DBLP:series/txtcs/Jukna11} says that either there is a matching of $\sqrt{m}$ edges or a star of size $\sqrt{m}$ in any graph with $m$ edges.  Note that, we are using analogous result in terms of vertex pairs.}. If there are $5 \log n$ pairwise disjoint vertex pairs in $P'$, then each such vertex pair $\{x,y\}$  forms an edge in $G$ with probability at least $1/4$, independent of other such vertex pairs. Now consider the case when there are $5 \log n$ vertex pairs having a common vertex in $P'$. Let $o$ be the center vertex and $\{o,z_1\},\ldots\{o,z_{5 \log n}\}$ be the vertex pairs. By Remark~\ref{rem:assumption} (i) and (ii),  we can assume that either we do not know the label of any vertex in $\{o, z_1,\ldots,z_{5 \log n}\}$, or we know the label of only $o$, or we know the label of exactly one $z_i$.  In any case, by Observation~\ref{obs:prob-edge}, the probability that $\{o,z_i\}$ is an edge is at least $1/4$, independent of all other $\{o,z_j\}$'s.  
  
  Putting everything together, the probability that none of the pairs in $P'$ form an edge holds with a probability of at most $(3/4)^{5 \log n}=o(1)$. So, we are done with the proof of Equation~\ref{eqn:smallo}, and hence Claim~\ref{clm:inter1}. 
 \end{proof}

  \begin{cl}[Claim~\ref{clm:inter2} restated]
    For any good node in the decision tree $\cT$, the following holds:
   $$ \pr_{G \sim \nodist}(\mbox{\staralg reaches } v) \leq \pr_{G \sim \yesdist}(\mbox{\staralg reaches } v).  $$
   \end{cl}
\begin{proof}
By Remark~\ref{rem:yes-no}, we can generate a graph $G \sim \yesdist$ first, and from that, we can generate $G'\sim \nodist$. Let ${\bf G'}$ be the set of all graphs in $\nodist$ that can be generated from $G \sim \yesdist$. We refer the set ${\bf G'}$ as the \emph{corresponding} set of graphs of $G \sim \yesdist$ in $\nodist$. Recall that the unknown graph $G \sim \yesdist \cup \nodist$ such that both $G \sim \yesdist$ and $G \sim \nodist$ hold with probability $1/2$ each. So, we can consider generating the unknown graph $G$ as follows: 

First generate $G \sim \yesdist$. The unknown graph is $G$ itself with probability $1/2$, and a graph in the correspondence set ${\bf G'}$  with probability $1/2$, generated as described in Remark~\ref{rem:yes-no}. This implies the following observation.

\begin{obs}\label{obs:yes-no}
{For any $H \sim \yesdist$ and its corresponding set of graphs ${\bf H'}$ in $\nodist$}, $$\pr(G=H~|~G \sim \yesdist) =\pr (G \in {\bf H'}~|~G \sim \nodist).$$
\end{obs}

In this claim, we are considering a good node $v$. Now, first consider the term 
$$\pr_{G \sim \yesdist}(\mbox{\staralg ($G$) reaches $v$}).$$
Let ${\cal H}_{v,\mbox{yes}}$ be the set of graphs in $\yesdist$ such that the algorithm reaches node $v$ if the unknown graph $G \in {\cal H}_{v,\mbox{yes}}$. So,
\begin{align*}
&\pr_{G \sim \yesdist}(\mbox{\staralg ($G$) reaches $v$}) \\
&=\sum\limits_{H \in {\cal H}_{v,\mbox{yes}}} \pr(G=H~|~G \sim \yesdist)
= \sum\limits_{H \in {\cal H}_{v,\mbox{yes}}}  \pr(G \in {\bf H'}~|~G \sim \nodist).
\end{align*}
The last equality follows from Observation~\ref{obs:yes-no}. Taking ${\cal H}_{v}'=\bigcup\limits_{H \in {\cal H}_{v,\mbox{yes}}}{\bf H'}$, 
\begin{equation}\label{eqn:rhs}
\pr_{G \sim \yesdist}(\mbox{\staralg ($G$) reaches $v$}) 
= \sum\limits_{H' \in {\cal H}_v'} \pr(G = H'~|~G \sim \nodist).
\end{equation}

Now, we consider the term 
$$\pr_{G \sim \nodist}(\mbox{\staralg ($G$) reaches $v$}).$$
Let ${\cal H}_{v,\mbox{no}}$ be the set of graphs in $\nodist$ such that the algorithm reaches node $v$ if the unknown graph $G \in {\cal H}_{v,\mbox{no}}$. Let $(E_Q,V(E_Q),e, \ell_v)$ be the data structure when the algorithm reaches node $v$. By the definition of a good node (see Definition~\ref{defi:badnode}), the algorithm reaches  node $v$ if $G \in {\cal H}_{v,\mbox{no}}$ and the corresponding $C'\subseteq C$ does not intersect with $V(E_Q)$. So, 
\begin{equation*}
\pr_{G \sim \nodist}(\mbox{\staralg ($G$) reaches $v$}) 
\leq \sum\limits_{H' \in {\cal H}_{v,\mbox{no}}} \pr(G=H'~|~G \sim \nodist).
\end{equation*} 
Now, observe that the algorithm does not reach node $v$ (in the decision tree) if the unknown graph $G \notin {\cal H}_v \cup {\cal H}_v'$. That is, ${\cal H}_{v,\mbox{no}} \subseteq {\cal H}_{v}'$. So, 
\begin{equation}\label{eqn:lhs}
\pr_{G \sim \nodist}(\mbox{\staralg ($G$) reaches $v$}) 
\leq \sum\limits_{H' \in {\cal H}_v'} \pr(G = H'~|~G \sim \nodist).
\end{equation} 
 By Equations~\ref{eqn:rhs} and~\ref{eqn:lhs}, we are done with the proof of Claim~\ref{clm:inter2}.
\end{proof}
\newpage
\section{Some probability results} \label{sec:prelim}

\begin{lem}[Chernoff-Hoeffding bound~\cite{DubhashiP09}]
\label{lem:cher_bound1}
Let $X_1, \ldots, X_n$ be independent random variables such that $X_i \in [0,1]$. For $X=\sum\limits_{i=1}^n X_i$ and $\mu=\E[X]$, the followings hold for any $0\leq \delta \leq 1$.
$$ \pr(\size{X-\mu} \geq \delta\mu) \leq 2\exp{\left(-\frac{\mu \delta^2}{3}\right)}$$

\end{lem}
\begin{lem}[Chernoff-Hoeffding bound~\cite{DubhashiP09}]
\label{lem:cher_bound2}
Let $X_1, \ldots, X_n$ be independent random variables such that $X_i \in [0,1]$. For $X=\sum\limits_{i=1}^n X_i$ and $\mu_l \leq \E[X] \leq \mu_h$, the followings hold for any $\delta >0$.
\begin{description}
\item[(i)] $\pr \left( X \geq \mu_h + \delta \right) \leq \exp{\left(-\frac{2\delta^2}{n}\right)}$.
\item[(ii)] $\pr \left( X \leq \mu_l - \delta \right) \leq \exp{\left(-\frac{2\delta^2} { n}\right)}$.
\end{description}

\end{lem}

\begin{lem}[Chernoff bound for bounded dependency, see Theorem~2.1 from ~\cite{DBLP:journals/rsa/Janson04}]
\label{lem:depend:high_exact_statement}
Let $X_1,\ldots,X_n$ be random variables such that $a_i \leq X_i \leq b_i$ and $X=\sum\limits_{i=1}^n X_i$. Let $\cD$ be the \emph{dependent} graph, where $V(\cD) = \left\{X_{1},\ldots,X_{n} \right\}$ and $ E(\cD) = \left\{(X_i,X_j): \mbox{$X_i$ and $X_j$ are dependent}\right\}$. Then for any $\delta >0$, 
$$ \pr(\size{X-\E[X]} \geq \delta) \leq  2\exp{\left(-\frac{2\delta^2} { \chi^*(\cD)\sum\limits_{i=1}^{n}(b_i-a_i)^2}\right)},$$
where $\chi^*(\cD)$ denotes the \emph{fractional chromatic number} of $\cD$.
\end{lem}

The following lemma is a special case of the above when $X_i$'s are indicator random variable.
\begin{lem}[See Corollary 2.6 from ~\cite{DBLP:journals/rsa/Janson04}]
\label{lem:depend:high_prob}
Let $X_1,\ldots,X_N$ be indicator random variables such that $\pr(X_i=1)=p$ for each $i$ there are at most $d$ $X_j$'s on which an $X_i$ depends, where $0 < p < 1$. For $X=\sum\limits_{i=1}^N X_i$ and $\mu= \E[X]$, the followings hold for any $\delta >0$.
\begin{description}
\item[(i)] $\pr(X \geq \mu + \delta) \leq \exp{\left(-\frac{\delta^2\left(1-\frac{d+1}{4N}\right) }{ 2(d+1)\left(Np+\frac{\delta}{3}\right)}\right)}$,
\item[(ii)]  $\pr(X \leq \mu - \delta) \leq \exp{\left(-\frac{\delta^2\left(1-\frac{d+1}{4N}\right) }{ 2(d+1)Np}\right)}$,
\end{description} 
\end{lem}

\end{document}